\documentclass[a4paper,11pt,reqno]{amsart}

\usepackage[margin=2.7cm]{geometry}
\usepackage[svgnames]{xcolor}
\usepackage{tikz}
\usepackage{color}
\usepackage{amssymb}
\usepackage{float}
\usepackage{enumitem}
\usepackage{array}
\usepackage{longtable}
\usepackage{ctable} 
\usepackage{multirow}
\usepackage{mathrsfs}
\usepackage{caption}
\usepackage{footmisc}

\newcolumntype{C}[1]{>{\centering\arraybackslash}m{#1}}
\newcolumntype{L}[1]{>{\arraybackslash}m{#1}}
\allowdisplaybreaks

\theoremstyle{definition}

\newtheorem{theorem}{Theorem}[section]
\newtheorem{corollary}[theorem]{Corollary}
\newtheorem{lemma}[theorem]{Lemma}
\newtheorem{definition}[theorem]{Definition}
\newtheorem{proposition}[theorem]{Proposition}
\newtheorem{remark}[theorem]{Remark}
\newtheorem{notation}[theorem]{Notation}

\newtheorem{example}[theorem]{Example}

\def\K{\mathbb{K}}
\def\F{\mathbb{F}}
\def\Fq{\mathbb{F}_q}
\def\C{\mathcal{C}}
\def\A{\mathcal{A}}
\def\B{\mathcal{B}}
\def\mD{\mathcal{D}}
\def\W{\mathcal{W}}
\def\Z{\mathbb{Z}}
\def\mS{\mathcal{S}}
\def\mH{\mathcal{H}}
\def\mG{\mathcal{G}}
\def\<{\left<}
\def\>{\right>}
\def\trk{\textup{trk}}

\def\maxrk{\textup{maxrk}}
\def\rk{\textup{rk}}
\newcommand{\qbin}[2]{\begin{bmatrix}{#1}\\ {#2}\end{bmatrix}_q}
\def\colsp{\textup{colsp}}
\def\rowsp{\textup{rowsp}}
\def\oB{\overline{B}}
\def\omA{\overline{\mathcal{A}}}
\newcommand{\dimq}[1]{\dim_{\Fq}\left(#1\right)}
\def\oW{\overline{W}}
\def\omW{\overline{\W}}
\def\mP{\mathscr{P}}
\def\ps{\textup{ps}}
\def\cl{\textup{cl}}
\def\D{\textup{D}}
\def\R{\textup{R}}
\def\Q{\mathbb{Q}}
\def\mL{\mathscr{L}}
\def\mM{\mathscr{M}}

\definecolor{myyellow}{RGB}{218,170,0}
\definecolor{myblue}{RGB}{0,75,135}
\definecolor{mygreen}{RGB}{61,174,43}

\DeclareCaptionLabelSeparator{none}{ }
\title{Tensor Codes and Their Invariants}

\author[Eimear Byrne]{Eimear Byrne}
\address{School of Mathematics and Statistics, University College Dublin, Belfield, Ireland}
\curraddr{}
\email{ebyrne@ucd.ie}
\thanks{}

\author[Giuseppe Cotardo]{Giuseppe Cotardo$^*$}
\address{School of Mathematics and Statistics, University College Dublin, Belfield, Ireland}
\curraddr{}
\email{giuseppe.cotardo@ucdconnect.ie}
\thanks{$^*$ The author was supported by the Irish Research Council, grant n. GOIPG/2018/2534.}



\begin{document}
\maketitle

\setlength{\abovedisplayskip}{5pt}
\setlength{\belowdisplayskip}{5pt}

\begin{abstract}
    In 1991, Roth introduced a natural generalization of rank metric codes, namely tensor codes. The latter are defined to be subspaces of $r$-tensors where the ambient space is endowed with the tensor rank as a distance function. In this work, we describe the general class of tensor codes and we study their invariants that correspond to different families of anticodes. In our context, an anticode is a {\em perfect space} that has some additional properties.  A perfect space is one that is spanned by tensors of rank 1.  Our use of the anticode concept is motivated by an interest in capturing structural properties of tensor codes. In particular, we indentify four different classes of tensor anticodes and show how these gives different information on the codes they describe.
    We also define the {\em generalized tensor binomial moments} and the {\em generalized tensor weight distribution} of a code and establish a bijection between these invariants. We use the generalized tensor binomial moments to define the concept of an $i$-tensor BMD code, which is an extremal code in relation to an inequality arising from them. Finally, we give MacWilliams identities for generalized tensor binomial moments.
\end{abstract}

\section{Introduction}

There are well established connections between linear block codes and 3-tensors, in particular the bilinear complexity of a 3-tensor can be related to the parameters of an associated linear code \cite{burgisser1996algebraic}. More recently, matrix codes have been studied in the context of 3-tensors \cite{byrne2021bilinear,byrne2019tensor}.
In \cite{roth1991maximum,roth1996tensor}, Roth studied {\em tensor codes}, that is, linear spaces of $r$-tensors in $\Fq^{n}\otimes\cdots\otimes\Fq^{n}$, where the ambient space is endowed with the tensor rank as a distance function. Such codes are natural generalizations of rank metric codes.
In \cite{byrne2020rank}, new invariants of coding theory were introduced, namely the $i$-th generalized rank weights. Matrix codes that are extremal in respect of a related bound, called $i$-binomial moment determined ($i$-BMD), all have the same $i$-th generalized rank weights. These new invariants can be used to distinguish inequivalent codes.
In \cite{ravagnani2016generalized}, the generalized rank weights were defined using anticodes and a Wei-duality result was derived for linear matrix codes.

Building on these initial papers, we generalize in two directions. Firstly, we study the general class of tensor codes and secondly we consider different anticodes related to the tensor rank and their corresponding invariants. In our context, an anticode is a {\em perfect space} that has some additional properties. A perfect space is one that is spanned by tensors of rank $1$. We identify four classes of different anticodes, namely the perfect spaces and three proper subtypes called {\em closure-type}, {\em Delsarte-type} and {\em Ravagnani-type} anticodes. We associate with each anticode type a sequence of generalized tensor weights. We remark that for the case $r=2$, i.e. for which the tensors codes are matrix codes, both the Delsarte-type and Ravagnani-type anticodes coincide and are the same as those defined in \cite{ravagnani2016generalized}, while the $j$-th generalized tensor weights corresponding to the perfect spaces are the the same as those defined in \cite{byrne2019tensor}. 

Our use of the anticode concept is motivated by an interest in capturing structural properties of tensor codes. As such, we define these objects in terms of tensor products. What we observe is a gradation in information about the parameters and properties of tensor codes. For example, the first generalized tensor weight associated with the {\em perfect spaces} corresponds exactly to the tensor rank distance of the code while the $k$-th generalized tensor weight associated with the perfect spaces corresponds to the tensor rank of the code itself as an $(r+1)$-tensor, where $k$ denotes the dimension of the code. These invariants give the most information on the code. As we progress through the types, from perfect spaces, through the closure-type to the Delsarte and Ravagnani-type anticodes and their associated invariants, we observe a loss of information about the code but we see decreasing complexity of computing these invariants, with the most noticeable difference occurring between the invariants for the perfect spaces and those of the closure-type anticodes. On the other hand, the Ravagnani-type anticodes are the only ones that exhibit a natural duality, similar to what can be observed in the matrix case ($r=2$). We compare these invariants through numerous examples. 

We remark on some important motivations on this topic. There are some significant departures from the matrix case for codes that are linear spaces of $r$-tensors for $r\geq 3$. Firstly, the `usual' Singleton bound is not sharp in general. While a strict improvement on this bound is obtained in \cite{roth1991maximum} for the codes in $\F_q^{n\times n \times n}$, only asymptotic bounds for higher order tensors are known. Moreover, it is not known if extremality with respect to this bound is an invariant of code duality. Compare this to the case of MRD codes, whose duals are also MRD and whose rank weight distributions are determined. We remark further that invariants associated with the closure-type anticodes not only provide a good distinguisher between inequivalent tensor codes, they are also a new invariant in the case of matrix codes.

This paper is organised as follows. In Section 2 we give preliminary results. In Section 3 we define what is meant by a tensor code, extending the definition originally given in \cite{roth1991maximum}. We define the closure of a tensor code and describe its properties. In Section 4, we introduce the concept of a tensor anticode and define the perfect space anticodes, closure-type, Delsarte-type, and the Ravagnani-type anticodes. We characterize these families and show the relations between them. In Section 5 we define the $j$-th generalized tensor weight associated with each type of anticode, and their corresponding dual weights, should they exist. We describe their properties and show how these different invariants offer different information on the codes they describe. In Section 6 we define the generalized tensor binomial moments and the generalized tensor weight distribution of a code and establish a bijection between these invariants. We also establish MacWilliams identities for generalized binomial moments, from which those for generalized tensor weight distributions follow. Particular formulations of such identities for $r=2$ are derived in \cite{byrne2020rank,delsarte1978bilinear,ravagnani2016rank}. We use the generalized tensor binomial moments to define the concept of an $i$-tensor BMD code (an $i$-TBMD code), which is an extremal code in relation to an inequality arising from the binomial moments. We summarize the different characterisations and properties of the collections of anticodes considered in this paper in the appendix. 

\section{Preliminaries and Notation}

	Throughout the paper, we let $q$ be a prime power and let $\F$ be the tensor product space $\Fq^{n_1}\otimes\cdots\otimes\Fq^{n_r}$ for some integers $r,n_1,\ldots,n_r$ such that $r,n_1\geq 2$. We will assume that $n_1=\min\{n_1,\ldots,n_r\}$ and that $n_r=\max\{n_1,\ldots, n_r\}$. We let $n:=\prod_{i=1}^rn_i$. Finally, for a positive integer $i$ we denote by $[i]$ the set $\{1,\ldots,i\}$.

    \begin{definition}
        We say that $U\in\F$ is a \textbf{rank-}$1$ (or \textbf{simple}) \textbf{tensor} if it can be expressed as $u^{(1)}\otimes \cdots \otimes u^{(r)}$, for some $u^{(i)}\in\Fq^{n_i}$, $i\in\{1,\ldots,r\}$.
    \end{definition}
    \begin{definition}
        The \textbf{rank} of $X\in\F$ is defined to be 
        \begin{equation*}
            \rk(X):=\min\left\{R\in\Z:X=\sum_{s=1}^Ru_s^{(1)}\otimes\cdots\otimes u_s^{(r)}\right\}.
        \end{equation*}
        That is, $\rk(X)$ is the least integer $\ell$ such that $X$ can be written as sum of $\ell$ simple tensors.
    \end{definition}
	The rank induces a metric on $\F$, i.e. the function 
	$$:\F \times \F\longrightarrow \Z : (X,Y)\longmapsto\rk(X-Y),$$ is a distance function (see \cite{roth1996tensor} for further details). 
	
    It is well-known that if $\{x_1^{(i)},\ldots,x_{n_i}^{(i)}\}$ is a basis of $\Fq^{n_i}$, for any $i\in\{1,\ldots,r\}$, then a basis of $\F$ is
    \begin{equation*}
        \left\{\bigotimes_{i=1}^r x_{j_i}^{(i)}:1\leq j_i\leq n_i, 1\leq i\leq r\right\}.
    \end{equation*}
    
    In particular, we have $\dimq{\F}=\prod_{i=1}^r\dimq{\Fq^{n_i}}=n$. An $r$-tensor $X\in\F$ can be represented as the map
    \begin{equation*}
        X: [n_1]\times\cdots\times[n_r]\longmapsto \Fq
    \end{equation*}
    given by $X=(X_{j_1,\ldots,j_r}:1\leq j_i\leq n_i, 1\leq i\leq r)$. Therefore, the map 
    \begin{equation*}
        :\F\longrightarrow\Fq^{n_1\times\cdots\times n_r}:X=\sum_{s=1}^{\rk(X)}\bigotimes_{i=1}^r x_s^{(i)}\longmapsto \left(X_{j_1,\ldots,j_r}=\sum_{s=1}^{\rk(X)}x_{j_i,s}^{(i)}\right)_{1\leq j_i\leq n_i, 1\leq i\leq r}
    \end{equation*}
   where, $x_s^{(i)}=(x_{j_i,s}^{(i)}:1\leq j_i\leq n_i)$ for any $i\in[r]$,  is an $\Fq$-isomorphism. In particular, we can identify $\F$ with $\Fq^{n_1\times\cdots\times n_r}$ and represent an $r$-tensor as an $r$-dimensional array.
   
   We briefly recall some basic definitions on posets and lattices (a standard reference is \cite[Chapter~3]{stanley2011enumerative}).  A \textbf{partially ordered set} (or \textbf{poset}) is a pair $(\mP,\leq)$ where $\mP$ is a non-empty set and $\leq $ is an binary relation satisfying the three axioms of \textit{reflexivity}, \textit{antisymmetry} and \textit{transitivity}. With an abuse of notation, in the following we denote by $\mP$ the poset $(\mP,\leq)$. Let $x,y\in\mP$, we say that $x$ and $y$ are \textbf{comparable} if $x\leq y$ or $y\leq x$, otherwise we say that they are \textbf{incomparable}. Moreover, we  use the standard notation $x<y$ to mean $x\leq y$ and $x\neq y$. An \textbf{upper bound} of $x$ and $y$ is an element $u\in\mP$ satisfying $x\leq u$ and $y\leq u$. A \textbf{least upper bound} (or \textbf{join}) of $x$ and $y$ is the element $u\in\mP$ such that $u\leq v$ for any other upper bound $v$ of $x$ and $y$. If such element exists then it is unique and it is denoted by $x\vee y$. Dually, we can define the \textbf{greatest upper bound} (or \textbf{meet}) $x\wedge y$ of $x$ and $y$, if it exists. If $\mP_1,\ldots,\mP_r$ are posets then the (\textbf{direct}) \textbf{product} of $\mP_1,\ldots,\mP_r$ is defined to be the poset $\mP_1\times\ldots\times\mP_r$ on the set $\{(x_1,\ldots,x_r):x_i\in\mP_i, i\in [r]\}$ such that $(x_1,\ldots,x_r)\leq (y_1,\ldots,y_r)$ in $\mP_1\times\ldots\times\mP_r$ if $x_i\leq y_i$ in $\mP_i$, for all $i\in[r]$. The \textbf{M\"obius function} of $\mL$ is defined recursively as
   \begin{equation*}
        \mu_\mP(x,y)=
        \begin{cases}
            1 & \textup{ if } x=y,\\
            \displaystyle-\sum_{x\leq z< y}\mu_{\mP}(x,z) & \textup{ if } x<y,\\
            0 & \textup{ otherwise}.
       \end{cases}
   \end{equation*}
   
      \begin{remark}
       One can check that the M\"obius function of the subspace poset $\mP$ of a vector spaces over a field $\K$ depends only on the dimension of their elements. In particular, for any $x_1,x_2,y_1,y_2\in\mP$ such that $x_1\leq y_1$, $x_2\leq y_2$, $\dim_\K(x_1)=\dim_\K(x_2)$ and $\dim_\K(y_1)=\dim_\K(y_2)$, we have $ \mu_\mP(x_1,y_1)= \mu_\mP(x_2,y_2)$. In the following, we use the notation $\mu_\mP(\dim_\K(x_1),\dim_\K(y_1))$ to mean $\mu_\mP(x_1,y_1)$.
   \end{remark}
   
   \begin{proposition}[M\"obius Inversion Formula]
       Let $\mP$ a poset such that for every element $z\in\mP$ the \textit{principal order ideal} generated by $z$, i.e. the set $\{x\leq z:x\in\mP\}$, is finite. Let $f,g:\mP\longrightarrow\K$, where $\K$ is a field. For all $y\in\mP$, we have
       \begin{equation*}
           g(y)=\sum_{x\leq y}f(x)\qquad \textup{ if and only if }\qquad f(y)=\sum_{x\leq y}\mu(x,y)g(x).
       \end{equation*}
   \end{proposition}
   
   We say that the poset $\mP$ is \textbf{locally finite} if for all $x,y\in\mP$ the sets $\{z\in\mP:x\leq z\leq y\}$ and $\{z\in\mP:x< z< y\}$ are finite (or empty).
   
   \begin{proposition}[The Product Theorem]
       Let $\mP_1\times\ldots\times\mP_r$ be the product of the locally finite posets $\mP_1,\ldots,\mP_r$. If $(x_1,\ldots,x_r)\leq(y_1,\ldots,y_r)$ in $\mP_1\times\ldots\times\mP_r$ then
       
       \begin{equation*}
           \mu_{\mP_1\times\ldots\times\mP_r}((x_1,\ldots,x_r),(y_1,\ldots,y_r))=\prod_{i=1}^r\mu_{\mP_i}(x_i,y_i).
       \end{equation*}
   \end{proposition}
   
   A \textbf{lattice} $\mL$  is a poset that contains the join and the meet of every pair of its elements. Every finite lattice has a minimum and a maximum element, denoted by $0_\mL$ and $1_\mL$ respectively. A subset $\mM$ of $\mL$ is a \textbf{sublattice} if is closed under the operations of $\vee$ and $\wedge$ in $\mL$. 
   Finally, we recall some well-known  properties of $q$-binomial coefficients (the reader is referred to \cite{andrews1998theory} for more details).
	
	\begin{definition}
		Let $a,b$ be integers. The $q$-binomial coefficients of $a$ and $b$ is
		\begin{equation*}
			\qbin{a}{b}=
			\begin{cases}
				0 & \textup{ if } b<0 \textup{ or } 0\leq a <b,\\[1ex]
				1 & \textup{ if } b=0 \textup{ and } a\geq 0,\\[1ex]
				\displaystyle\prod_{i=0}^{b-1}\frac{q^{a-i}-1}{q^{i+1}-1} & \textup{ if } b>0 \textup{ and } a\geq b,\\[1ex]
				\displaystyle(-1)^bq^{ab-\binom{b}{2}}\qbin{-a+b-1}{b} & \textup{ if } b>0 \textup{ and } a<0.
			\end{cases}
		\end{equation*}
	\end{definition}
	
	\begin{lemma}
		\label{lem:bin}
		Let $a,b,c$ be integers. The following hold.
		\begin{enumerate}
		\setlength\itemsep{1em}
			\item $\displaystyle \qbin{a}{b}\qbin{b}{c}=\qbin{a}{c}\qbin{a-c}{a-b}$.
			\item $\displaystyle \qbin{a+b}{c}=\sum_{j=0}^cq^{j(b-c+j)}\qbin{a}{j}\qbin{b}{c-j}=\sum_{j=0}^cq^{(c-j)(a-j)}\qbin{a}{j}\qbin{b}{c-j}$.
		\end{enumerate}
	\end{lemma}
	
	\section{Tensor codes}

  Roth introduced the class of tensor codes, as a generalisation of rank-metric codes, in \cite{roth1991maximum,roth1996tensor} and for the case $n_1=\cdots=n_r$. 
  We will consider this class of codes, which are subspaces of $\F$, without the constraint that the $n_i$ are all equal.
    
    \begin{definition}
        A \textbf{tensor code} is a $k$-dimensional subspace $\C$ of $\F$. We define the \textbf{maximum rank} $\maxrk(\C)$ as $\maxrk(\C):=\max\{\rk(C):C\in \C\}$. The \textbf{minimum} (\textbf{tensor rank}) \textbf{distance} of a non-zero code $\C$ is $d(\C):=\min\{\rk(C):C \in\C, c\neq 0\}$. The \textbf{tensor rank} $\trk(\C)$ of $\C$ is the minimum dimension of a perfect space containing $\C$.
    \end{definition}
    
	We refer to $n_1\times\cdots\times n_r,k,d(\C)$ as the \textbf{code parameters} of $\C$ and we say that $\C$ is an $\Fq$-$[n_1\times\cdots\times n_r,k,d(\C)]$ code. Throughout the paper, we let $\C$ be an $\Fq$-$[n_1\times\cdots\times n_r,k,d(\C)]$ code and we write $d$ instead of $d(\C)$. For any $i$, we denote by $\cdot$ the dot product on $\Fq^{n_i}$, defined in the usual way as $x\cdot y = \sum_{j=1}^{n_i} x_jy_j$ for $x,y \in \Fq^{n_i}$. 
	We furthermore denote by $*$ the non-degenerate bilinear form defined by:
    \begin{equation*}
        \begin{array}{cccc}
              *:&\F\times\F &\longrightarrow&\Fq\\
              &\displaystyle\left(\bigotimes_{i=1}^ru^{(i)},\bigotimes_{i=1}^rv^{(i)}\right)&\longmapsto&\displaystyle\prod_{i=1}^r\left(u^{(i)}\cdot v^{(i)}\right).
        \end{array}
    \end{equation*}

    \begin{remark}
        One can check that, for all $X,Y\in\F$, we have
        \begin{equation*}
            X*Y=\sum_{j_1=1}^{n_1}\cdots\sum_{j_r=1}^{n_r}X_{j_1,\cdots,j_r}Y_{j_1,\cdots,j_r}.
        \end{equation*}
    \end{remark}

    \begin{definition}
      We define the \textbf{dual} of $\C$ to be
        \begin{equation*}
            \C^\perp:=\{X\in\F:X*C=0 \textup{ for all } C\in\C\}.
        \end{equation*}
    \end{definition}
    
	We summarize some basic facts on the dual tensor code.
    
    \begin{lemma}
    \label{lem:propdual}
        Let $\C,\mD\leq\F$ be codes. The following hold.
        \begin{enumerate}
         \setlength\itemsep{0.5em}
            \item\label{item1:propdual} $\left(\C^\perp\right)^\perp=\C$.
            \item\label{item2:propdual} $\dimq{\C^\perp}=n-\dimq{\C}$.
            \item\label{item3:propdual} $(\C\cap\mD)^\perp=\C^\perp+\mD^\perp$.
        \end{enumerate}
    \end{lemma}
    
    In the remainder, unless explicitly stated, $\C\leq\F$ will denote a tensor code of dimension $k$ over $\Fq$ and minimum distance $d$. We now define the {\em closure} of $\C$. As the reader will see, this corresponds to a minimal {\em anticode} containing $\C$. 
    Before giving a formal definition, let us state what the closure is for a matrix code.
    To this end, let $r=2$. For each $C \in \C$, denote by $\rowsp(C)$ and $\colsp(C)$ the vector spaces generated by the rows and the columns of $C$, respectively, 
    and define the closure of $\<C\>_{\Fq}$ to be $\colsp(C) \otimes \rowsp(C)$.
    More generally, we define $\cl(\C)=\colsp(\C)\otimes\rowsp(\C)$, where 
    \begin{equation*}
        \rowsp(\C)=\sum_{C\in\C}\rowsp(C) \qquad\textup{ and }\qquad \colsp(\C)=\sum_{C\in\C} \colsp(C).
    \end{equation*}
    We remark that the closure of a matrix code does not require knowledge of its tensor rank. However, in our first definition of the closure, we will express it in terms of the decomposition of codewords as rank-1 tensors.  
    
	\begin{definition}
	\label{def:closure}
	    Let $ C=\sum_{j=1}^{\rk(C)}c_j^{(1)}\otimes\cdots\otimes c_j^{(r)}\in\C$ be a codeword. For each $i\in\{1,\ldots,r\}$, we define the subspace
	    $\cl(C)^{(i)}:=\<c_j^{(i)}:j\in\{1,\ldots,\rk(C)\}\>_{\Fq} \leq \Fq^{n_i}$. The \textbf{closure} of $\C$ is the tensor space
	    \begin{equation*}
	        \cl(\C):=\bigotimes_{i=1}^r\cl(\C)^{(i)} \qquad \textup{ where }\qquad \cl(\C)^{(i)}:=\sum_{C\in \C}\cl(c)^{(i)}
	    \end{equation*}
	    for any $i\in\{1,\ldots,r\}$. For a tensor $X\in\F$ we write $\cl(X)$ to mean $\cl\left(\<X\>_{\Fq}\right)$.
	\end{definition}
    
   Clearly, every tensor code is contained in its own closure.
   
   \begin{definition}
       A \textbf{slice} of a tensor $X\in\F$ is a two-dimensional fragment, obtained by fixing all the indices but two of $X$ seen as $r$-dimensional array.
   \end{definition}
    
    Figures \ref{fig:slice1}, \ref{fig:slice2} and \ref{fig:slice3} show the different types of slices of a $3$-tensor obtained by fixing the first, second and third index respectively. 
   
   \begin{figure}[htb]
   \centering
       \begin{minipage}{0.32\textwidth}
       \centering
           \includegraphics[width=0.8\linewidth]{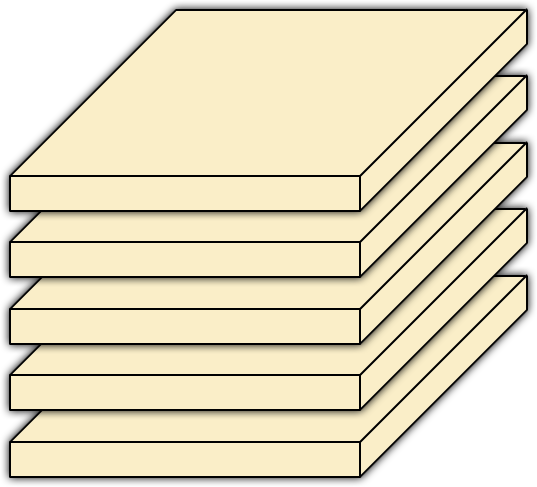}
            \caption{\label{fig:slice1} Index-1 }
       \end{minipage}
       \begin{minipage}{0.32\textwidth}
       \centering
           \includegraphics[width=0.8\linewidth]{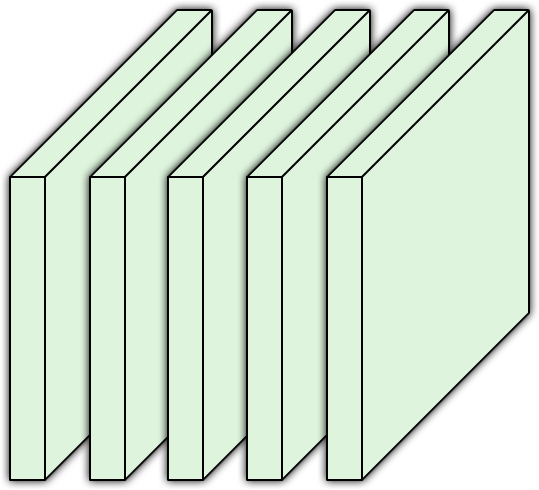}
            \caption{\label{fig:slice2} Index-2 }
       \end{minipage}
       \begin{minipage}{0.32\textwidth}
       \centering
           \includegraphics[width=0.8\linewidth]{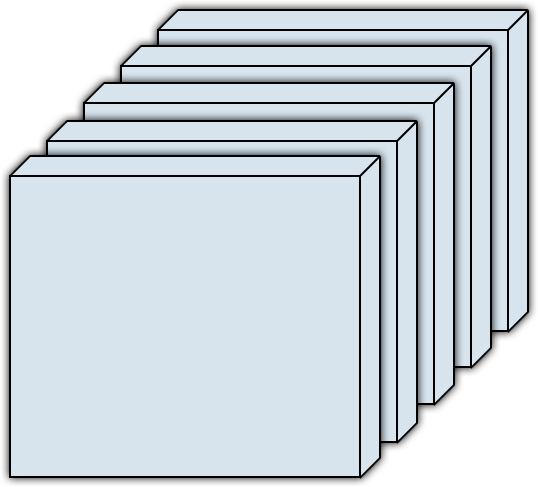}
           \caption{\label{fig:slice3} Index-3 }
       \end{minipage}
   \end{figure}
   
   The following provides a convenient representation of $3$-tensors.
   
   \begin{notation}
       Let $s\in\{1,\ldots,n_3\}$ and we define the following projection map
       \begin{equation*}
           \begin{array}{ccccc}
                \Sigma_s &: & \Fq^{n_1}\otimes\Fq^{n_2}\otimes\Fq^{n_3} &\longrightarrow &  \Fq^{n_1\times n_2} \\[0.5em]
                &&X & \longmapsto &\left(X_{t_1,t_2,s}:t_1\in[n_1], t_2\in[n_2]\right).
           \end{array}
       \end{equation*}
       In the following, we represent a $3$-tensor $X\in  \Fq^{n_1}\otimes\Fq^{n_2}\otimes\Fq^{n_3}$ as $1$-dimensional array whose components are the slices of $X$ obtained by fixing the third index, i.e. 
       \begin{equation*}
           X=\left(\Sigma_1(X)\mid\Sigma_2(X)\mid\cdots\mid\Sigma_{n_3}(X)\right).
       \end{equation*}
       We use to refer to this representation as the \textbf{matrix representation} of $X$. 
   \end{notation}
   
   \begin{example}
   \label{ex:closure}
    Let $\C\leq \F_3^2\otimes \F_3^3\otimes \F_3^4$ be the $3$-dimensional tensor code generated by
    \begin{align*}
        &X:=\left(\begin{array}{ccc|ccc|ccc|ccc}
            2&1&1&2&1&0&0&0&1&2&1&0\\
            0&0&0&0&0&0&0&0&0&0&0&0
        \end{array}\right),\\
        &Y:=\left(\begin{array}{ccc|ccc|ccc|ccc}
            0&0&0&0&0&0&0&0&0&0&0&0\\
            1&2&1&1&1&0&2&2&0&2&0&1
         \end{array}\right),\\
        &Z:=\left(\begin{array}{ccc|ccc|ccc|ccc}
            2&0&2&1&0&2&1&0&1&2&0&2\\
            2&0&1&0&0&0&1&0&2&2&0&1
        \end{array}\right).
    \end{align*}
    
    One can observe that $\rk(X)=2$ and, in particular, $X$ is the sum of the following rank-1 tensors
    \begin{align*}
        &\left(\begin{array}{ccc|ccc|ccc|ccc}
            2&1&1&2&1&0&0&0&0&2&1&0\\
            0&0&0&0&0&0&0&0&0&0&0&0
        \end{array}\right)=(1,0)\otimes(2,1,0)\otimes(1,1,0,1),\\
        &\left(\begin{array}{ccc|ccc|ccc|ccc}
            0&0&0&0&0&0&0&0&1&0&0&0\\
            0&0&0&0&0&0&0&0&0&0&0&0
        \end{array}\right)=(1,0)\otimes(0,0,1)\otimes(1,0,1,0).
    \end{align*}
    
    Therefore, we have $\cl(X)=\<(1,0)\>_{\F_3}\otimes\<(2,1,0),(0,0,1)\>_{\F_3}\otimes\<(1,1,0,1),(0,1,2,1)\>_{\F_3}$. Analogously, one can check that
    \begin{itemize}
    \setlength\itemsep{0.5em}
        \item $\cl(Y)=\<(0,1)\>_{\F_3}\otimes \<(1,0,2),(0,1,1)\>_{\F_3}\otimes\<(1,0,0,1),(0,1,2,1)\>_{\F_3}$,
        \item  $\cl(Z)=\<(1,0),(0,1)\>_{\F_3}\otimes \<(1,0,0),(0,0,1)\>_{\F_3}\otimes\<(1,0,2,1),(0,1,0,0)\>_{\F_3}$.
    \end{itemize}
    Finally, one can easily check that
    \begin{align*}
        \cl(\C)^{(1)}&=\cl(X)^{(1)}+\cl(Y)^{(1)}+\cl(Z)^{(1)}=\<(1,0)\>_{\F_3}+\<(0,1)\>_{\F_3}+\<(1,0),(0,1)\>_{\F_3} = \F_3^2,\\[1em]
        \cl(\C)^{(2)}&=\cl(X)^{(2)}+\cl(Y)^{(2)}+\cl(Z)^{(2)}\\&=\<(2,1,0),(0,0,1)\>_{\F_3}+\<(1,0,2),(0,1,1)\>_{\F_3}+\<(1,0,0),(0,0,1)\>_{\F_3} =\F_3^3,\\[1em]
        \cl(\C)^{(3)}&=\cl(X)^{(3)}+\cl(Y)^{(3)}+\cl(Z)^{(3)}\\&=\<(1,1,0,1),(0,1,2,1)\>_{\F_3}+\<(1,0,0,1),(0,1,2,1)\>_{\F_3}+\<(1,0,2,1),(0,1,0,0)\>_{\F_3}=\F_3^4.
    \end{align*}
    Therefore, $\cl(\C)=\F_3^2\otimes \F_3^3\otimes \F_3^4$.
\end{example}
   
   We recall the following well-known definition from tensors.
   
   \begin{definition}
       A \textbf{fiber} of a tensor $X\in\F$ is a one-dimensional fragment, obtained by fixing all the indices but one of $X$ seen as $r$-dimensional array.
   \end{definition}
   
  Figures \ref{fig:row}, \ref{fig:col} and \ref{fig:tube} show the different types of fibers of a $3$-tensor. 
  
   \begin{figure}[htb]
       \begin{minipage}{0.32\textwidth}
       \centering
           \includegraphics[width=0.8\linewidth]{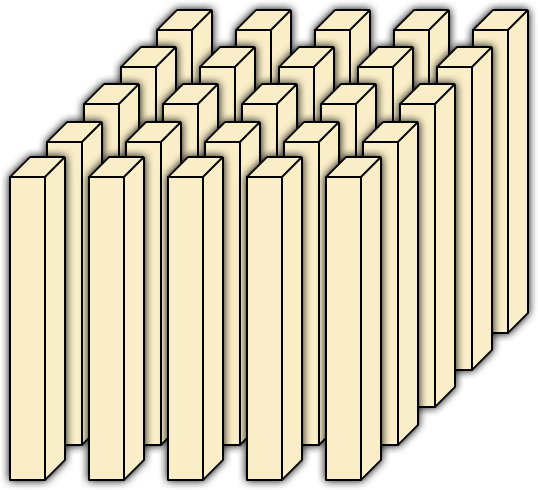}
            \caption{\label{fig:col} Column }
       \end{minipage}
       \begin{minipage}{0.32\textwidth}
       \centering
           \includegraphics[width=0.8\linewidth]{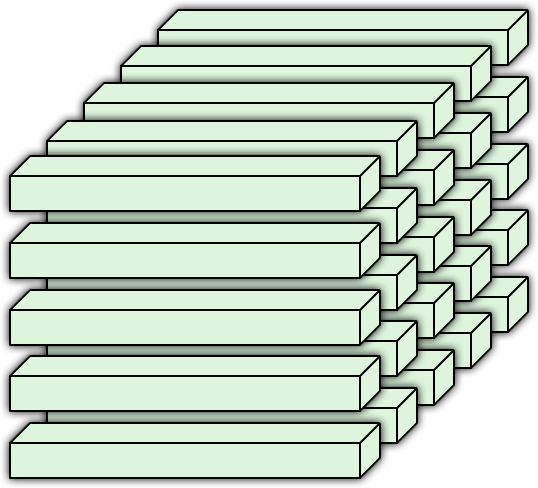}
           \caption{\label{fig:row} Row }
       \end{minipage}
       \begin{minipage}{0.32\textwidth}
       \centering
           \includegraphics[width=0.8\linewidth]{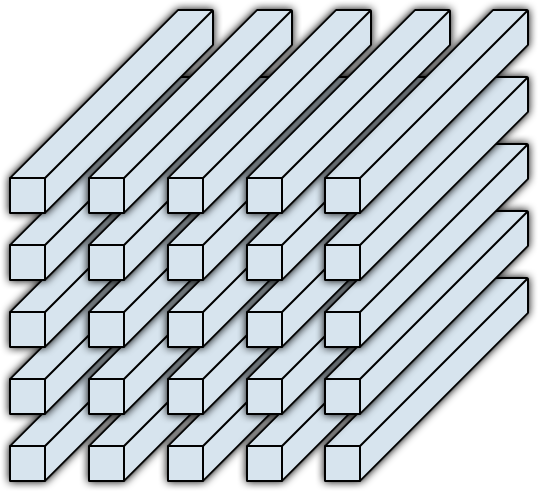}
            \caption{\label{fig:tube} Tube }
       \end{minipage}
   \end{figure}
   
    The following definition provide an explicit projection map that returns the different fibers of an $r$-tensor. 
    
    \begin{notation}
        For any integer $i\in\Z$, we denote by $\overline{i}$ the unique representative of the equivalence class of $i \equiv r$ in the set $\{1,\ldots,r\}$.  
    \end{notation}
    
    \begin{definition}
    \label{def:project}
        Let $i\in\{1,\ldots,r\}$ be an integer and let $s=(s_1,\ldots,s_r) \in [n_{\overline{i+1}}]\times \cdots\times [n_{\overline{i+r}}]$. For each positive integer $t$, define
        \begin{equation*}
             \sigma_1(s,t) = (t,s_1,\ldots,s_{r-1}),\;\; \sigma_2(s,t) = (s_{r-1},t,s_1,\ldots,s_{r-2}),\;\;\ldots\;\;,\;\;\sigma_r(s,t) = (s_{1},\ldots,s_{r-1},t).
        \end{equation*}
        We define the following projection map 
        \begin{equation*}
            \begin{array}{cccc}
                \Pi_{i,s}:&\Fq^{n_1}\otimes\ldots\otimes\Fq^{n_r}&\longrightarrow&\Fq^{n_i}\\
                &X &\longmapsto &(X_{\sigma_i(s,t)}:t\in[n_i]).
            \end{array}
        \end{equation*}
    \end{definition}
    
    For any $X\in\F$, we refer to  $\{\Pi_{i,s}:s\in [n_{\overline{i+1}}]\times \cdots\times [n_{\overline{i+r}}]\}$ as the set of \textbf{mode}-$i$ fibers of $X$. In particular if $X$ is a $3$-tensor then we have that column, row and tube fibers are mode-$1$, mode-$2$ and mode-$3$ fibers respectively. The next result follows from Definitions \ref{def:closure} and \ref{def:project}. Its purpose is to demonstrate that the closure of a tensor (and hence a tensor code) can be constructed without obtaining a decomposition of it as a sum of rank-$1$ tensors.
    
    \begin{lemma}
    \label{thm:charclos}
         The following hold. 
         \begin{enumerate}
             \item For any $i\in\{1,\ldots,r\}$ and $c \in \C$,
         $\displaystyle\cl(c)^{(i)}=\<\Pi_{i,s}(c):s \in [n_{\overline{i+1}}]\times \cdots\times [n_{\overline{i+r}}]\>_{\Fq}.$
         \item
            $
           \displaystyle \cl(\C)=\bigotimes_{i=1}^r\sum_{c\in\C}\<\Pi_{i,s}(c):s \in [n_{\overline{i+1}}]\times \cdots\times [n_{\overline{i+r}}]\>_{\Fq}.
        $
         \end{enumerate}
         
    \end{lemma}
    
    \begin{example}
        Let $X$ be as in Example \ref{ex:closure} and we have
        \begin{equation*}
        \begin{array}{cc}
        \renewcommand{\arraystretch}{1.5}
            \begin{array}{|ccc|}
            \hline
                 X_{1,1,1}=2 & X_{1,2,1}=1 & X_{1,3,1}=1\\
                X_{2,1,1}=0 & X_{2,2,1}=0 & X_{2,3,1}=0\\\hline
            \end{array} &
            \renewcommand{\arraystretch}{1.5}
            \begin{array}{|ccc|}
            \hline
                 X_{1,1,1}=2 & X_{1,2,1}=1 & X_{1,3,1}=1\\
                    X_{2,1,1}=0 & X_{2,2,1}=0 & X_{2,3,1}=0\\\hline
            \end{array}\\\\
            \renewcommand{\arraystretch}{1.5}
            \begin{array}{|ccc|}
            \hline
                X_{1,1,3}=0 & X_{1,2,3}=0 & X_{1,3,3}=1\\
                X_{2,1,3}=0 & X_{2,2,3}=0 & X_{2,3,3}=0\\\hline
            \end{array} &
            \renewcommand{\arraystretch}{1.5}
            \begin{array}{|ccc|}
            \hline
                X_{1,1,4}=2 &  X_{1,2,4}=1 &  X_{1,3,4}=0\\
                X_{2,1,4}=0 & X_{2,2,4}=0 & X_{2,3,4}=0\\\hline
            \end{array}
         \end{array}
        \end{equation*}
        We want to compute $\cl(X)^{(i)}$ for $i\in\{1,2,3\}$.
        \begin{itemize}
        \setlength\itemsep{1em}
            \item Let $i=1$. One can check that $\{\Pi_{1,s}(X):s \in [3]\times[4]\}$ is the set of vectors of $\F_3^2$ given by  $ \{(X_{1,s_1,s_2},X_{2,s_1,s_2}):s \in [3]\times[4]\}$, that is
            \begin{equation*}
                \{(2,0),(1,0),(0,0)\}.
            \end{equation*}
            \item Let $i=2$. One can check that $\{\Pi_{2,s}(X):s \in [4]\times [2]\}$ is the set of vectors of $\F_3^3$ given by $ \{(X_{s_2,1,s_1},X_{s_2,2,s_1},X_{s_2,3,s_1}):s \in [4]\times[2]\}$, that is
            \begin{equation*}
                \{(2,1,1),(2,1,0),(0,0,1),(0,0,0)\}.
            \end{equation*}
            \item Let $i=3$. One can check that $\{\Pi_{3,s}(X):s \in [2]\times [3]\}$ is the set of vectors of $\F_3^4$ given by $\{(X_{s_1,s_2,1},X_{s_1,s_2,2},X_{s_1,s_2,3},X_{s_1,s_2,4}):s \in [2]\times[3]\}$, that is 
            \begin{equation*}
                \{(2,2,0,2),(1,1,0,1),(1,0,1,0),(0,0,0,0)\}.
            \end{equation*}
        \end{itemize}
        Therefore, as a consequence of Lemma \ref{thm:charclos} we have 
        \begin{equation*}
            \cl(X)^{(1)}=\<(1,0)\>_{\F_3},\quad
            \cl(X)^{(2)}=\<(2,1,0),(0,0,1)\>_{\F_3},\quad
            \cl(X)^{(3)}=\<(1,1,0,1),(0,1,2,1)\>_{\F_3}
        \end{equation*}
        which implies $\cl(X)=\<(1,0)\>_{\F_3}\otimes\<(2,1,0),(0,0,1)\>_{\F_3}\otimes\<(1,1,0,1),(0,1,2,1)\>_{\F_3}$ as shown in Example \ref{ex:closure}.
    \end{example}
    
\section{Theory of Anticodes}
\label{sec:anticodes}
	Inspired by the work in \cite{ravagnani2016generalized}, in this section we develop the general theory of anticodes for tensor codes. We then define different classes of anticodes and we study their properties. In the following, we say that $\C$ of $\F$ is \textbf{perfect} if $\C$ is generated by simple tensors, (c.f. \cite{atkinson1983ranks}). 
	
	\begin{definition}
	\label{def:anticode}
		We say that $A\leq\F$ is a (\textbf{tensor}) \textbf{anticode} if $A$ is perfect.  If $\A$ is a collection of anticodes of $\F$ we write $\A_a$ to denote the set of all anticodes of $\A$ of dimension $a$, where $a\in\{0,\ldots,n\}$.
		If $\{A^\perp:A\in\A\}$ is a collection of anticodes, then we define the collection of \textbf{dual} (\textbf{tensor}) \textbf{anticodes} of $\A$ to be $\omA:=\{A^\perp:A\in\A\}$. 
	\end{definition}
	
    We now define the different classes of anticodes that we will consider in this paper.
	
	\begin{definition}\label{anticodedef}
	    We define the following.
	    \begin{enumerate}
	    \setlength\itemsep{0.5em}
	        \item The family of \textbf{perfect spaces} is  \\$\A^\ps:=\{A\leq\F:A \textup{ is perfect}\}$.
	        \item The family of \textbf{closure-type anticodes} is \\$\A^\cl:=\{A:A \in \A^\ps\mid A=\cl(A)\}$.
	        \item The family of \textbf{Delsarte-type anticodes} is 
	       \\$\A^\D:=\{A:A \in \A^\cl\mid A^\perp\in\A^\cl \textup{ and } A^{(p)}=\Fq^{n_p} \textup{ for some } p\in[r] \textup{ such that } n_p=n_r\}$.
	        \item The family of \textbf{Ravagnani-type anticodes} is
	            \\$\A^\R:=\{A:A \in \A^\cl\mid A^\perp\in\A^\cl \textup{ and } A^{(p)}=\Fq^{n_p} \textup{ for any } p\in[r] \textup{ such that }n_p \neq n_1\}$, \\if $n_1<n_r$, and 
	            \\$\A^\R:=\{A:A \in \A^\cl\mid A^\perp\in\A^\cl \textup{ and } A^{(p)}=\Fq^{n_p} \textup{ for some } p\in[r] \textup{ such that }n_p = n_1\}$,\\ if $n_1=n_2=\cdots =n_r$.
	    \end{enumerate}
	\end{definition}
	
	Clearly, $\A^R, \A^D \subset \A^{\cl} \subset \A^{\ps}.$
	The choice of these families of anticodes is motivated by the fact that they allow us to define invariants for tensor codes that extend the theory of invariants for matrix codes in the rank-metric. 
	For example, the definition of the \textit{generalized tensor ranks} of a matrix code as defined in \cite[Section~6]{byrne2019tensor}, relies on the notion of perfect spaces (the collection $\A^\ps$), although this was not explicitly stated there. Moreover, we will show that the Delsarte-type and the Ravagnani-type anticodes extend the theory of Delsarte optimal anticodes for rank-metric codes (see \cite{ravagnani2016generalized}). 
	
	\begin{notation}
	\label{notation:isom}
	    For each $i\in\{1,\ldots r\}$ and let $\phi_i:\Fq^{n_i}\longrightarrow\F_q^{n_i}$ be a map. We define
	    $$\phi:\F\longrightarrow\F:\bigotimes_{i=1}^rX^{(i)}\longmapsto\bigotimes_{i=1}^r\phi_i\left(X^{(i)}\right).$$
	    Clearly, if $\phi_i$ is an $\F_q$-isomorphism for each $i$ then $\phi$ is an $\F_q$-linear isometry with respect to the tensor rank. Let $1 \leq j<s\leq r$ and suppose that $n_j=n_s$. Define the map 
	    $$\tau_{j,s}:\F\longrightarrow\F:\bigotimes_{i=1}^rX^{(i)}\longmapsto\left(\bigotimes_{i=1}^{j-1}X^{(i)}\right)\otimes X^{(s)}\otimes\left(\bigotimes_{i=j+1}^{s-1} X^{(i)}\right)\otimes X^{(j)}\otimes \left(\bigotimes_{i=s+1}^{r} X^{(i)}\right).$$
	    That is, $\tau_{j,s}$ is the map that interchanges $X^{(j)}$ and $X^{(s)}$ in $\bigotimes_{i=1}^rX^{(i)}$. It is easy to see that $\tau_{j,s}$ is an $\F_q$-linear isometry of $\F$.
	\end{notation}
	
	We remark that $\tau_{j,s}$ corresponds to the transposition of matrices for the case $r=2$.
    
    \begin{remark}
        Observe that the set $\A^\ps$ of perfect spaces is closed under the isometries defined in Notation \ref{notation:isom}. As the reader will see in the remainder of this section, all the classes of anticodes defined in Definition \ref{def:anticode} are closed under these isometries. This fact will easily follow from the characterizations of the anticodes we now provide.
    \end{remark}
    
    The following example shows that the dual (the orthogonal complement) of a perfect space is not perfect in general. This implies that the set $\overline{\A^\ps}$ is not defined.
    
    \begin{example}
    \label{ex:dualperfsp}
        Let $A\leq\F_3^2\otimes\F_3^3$ be the perfect space of dimension $5$ generated by 
        \begin{equation*}\left\{
            \begin{pmatrix}0&0&0\\0&1&1\end{pmatrix},
            \begin{pmatrix}1&2&1\\1&2&1\end{pmatrix},
            \begin{pmatrix}0&0&1\\0&0&1\end{pmatrix},
            \begin{pmatrix}1&0&2\\0&0&0\end{pmatrix},
            \begin{pmatrix}1&0&1\\0&0&0\end{pmatrix}
            \right\}.
        \end{equation*}
        One can check that the dual of $\C$ is the $1$-dimensional space 
        \begin{equation*}
            A^\perp=\<\begin{pmatrix}0&1&0\\1&0&0\end{pmatrix}\>_{\F_3}
        \end{equation*}
        which is clearly not perfect.
    \end{example}
    
    \begin{remark}
        Observe that $\A^\ps$ is a lattice where the partial order relation is given by the inclusion of subspaces. Denote by $\vee_\ps$ and $\wedge_\ps$ the operations of join and meet, respectively, in $\A^\ps$. More in detail, let $A,B\in\A^\ps$, we have that $A\vee_\ps B$ is the vector space sum $A+B$ and $A\wedge_\ps B$ is the space $\<X\in A\cap B:\rk(X)=1\>_{\Fq}$. Clearly, $A\vee_\ps B$ and $A\wedge_\ps B$ are both perfect spaces. 
    \end{remark}
	
	We now characterize the families of anticodes $\A^\cl$, $\A^\D$, and $\A^\R$. The following result is an immediate consequence of the definition of the closure of a tensor code (Definition \ref{def:closure}) and the fact that the tensor product of a collection of vector spaces is perfect. 
	
    \begin{proposition}
    \label{prop:charClosure}
       The following holds.
       \begin{equation*}
           \A^\cl=\left\{\bigotimes_{i=1}^rA^{(i)}:A^{(i)}\leq\Fq^{n_i}\textup{ for all } i \in \{1,\ldots,r\}\right\}.
       \end{equation*}
    \end{proposition}
    
    \begin{notation}
       Let $i\in\{1,\ldots, r\}$. We denote by $\mL(\Fq^{n_i})$ the lattice of subspaces of $\Fq^{n_i}$ where, for any subspaces $A,B$ of $\Fq^{n_i}$, we define the join of $A$ and $B$ to be $A+B$ and their meet to be $A\cap B$. For any $A\in\A^\cl$ we denote by $A$ the unique subspace of $\Fq^{n_i}$ such that $A=\bigotimes_{i=1}^rA^{(i)}$.
    \end{notation}
    
    \begin{lemma}
    \label{lem:BleqA}
        Let $A,B\in\A^\cl$. We have $B\leq A$ if and only if, for all $i\in\{1,\ldots,r\}$,  $B^{(i)}\leq A^{(i)}$ in $\mL(\Fq^{n_i})$.
    \end{lemma}
    \begin{proof}
        Since $A,B\in\A^\cl$ then $A=\bigotimes_{i=1}^r A^{(i)}$ and $B=\bigotimes_{i=1}^r B^{(i)}$ for some $A^{(i)},B^{(i)}\leq\Fq^{n_i}$ and $i\in\{1,\ldots,r\}$. Clearly, if $B^{(i)}\leq A^{(i)}$ in $\mL(\Fq^{n_i})$ for all $i\in\{1,\ldots,r\}$ then $B\leq A$. It remains to prove the other implication. Suppose $B\leq A$ and observe that $B$ has a basis of simple tensors $C:=\bigotimes_{i=1}^rc^{(i)}$ where, for any $i\in\{1,\ldots,r\}$, $c^{(i)}$ is an element of a basis for $B^{(i)}$. Since $B\leq A$ then $c$ is also an element of $A$. This implies that, for any  $i\in\{1,\ldots,r\}$, $c^{(i)}$ is also an element of $A^{(i)}$. The statement follows.
    \end{proof}
    
    \begin{proposition}
    \label{thm:latticecl}
        $\A^\cl$ is a sublattice of $\A^\ps$.
    \end{proposition}
    \begin{proof}
        We claim that $A\wedge_\ps B=A\cap B$. 
        It is sufficient to show that $A\cap B\in\A^\cl$. We have 
        \begin{equation*}
           A\cap B=\<\bigotimes_{i=1}^rx^{(i)}:x^{(i)}\in A^{(i)} \textup{ and } x^{(i)}\in B^{(i)} \textup{ for all } i\in[r]\>_{\Fq}= \bigotimes_{i=1}^r\left(A^{(i)}\cap B^{(i)}\right)
        \end{equation*}
        which shows that $A\cap B$ is an element of $\A^\cl$. Therefore $\A^\cl$ is closed under the operation of $\wedge_\ps$. It is easy to see that for any $A,B\in A^\cl$, we have $A\vee_\ps B=A+B$ and that the latter is a closure-type anticode. This shows that $\A^\cl$ is also closed under the operation of $\vee_\ps$ and therefore $\A^\cl$ is a sublattice of $\A^\ps.$
    \end{proof}
    
	\begin{corollary}
       \label{thm:mincl}
       The closure of $\C$ is the minimal element of $\A^\cl$ containing $\C$, that is,
       \begin{equation*}
           \cl(\C)=\bigcap\{A\in \A^\cl:\C\leq A\}.
       \end{equation*}
    \end{corollary}
    \begin{proof}
      	First of all, observe that $\cl(\C)=\bigotimes_{i=1}^r\cl(\C)^{(i)}\in\A^\cl$ by definition of the closure. If $\C\in\A^\cl$ then the result is immediate. We assume $\C\notin\A^\cl$ in the remainder of the proof. Suppose that there exists a closure-type anticode $A\in\A^\cl$ such that $\C< A<\cl(\C)$, write $A=\bigotimes_{i=1}^rA^{(i)}$. By Lemma \ref{lem:BleqA}, there exists $u\in \cl(\C)^{(i)}$ such that $u\notin A^{(i)}$, for some $i\in\{1,\ldots,r\}$. This leads to the contradiction 
      	\begin{equation*}
      	    A\cap \C=\left\{C\in\C: c^{(i)} \notin \<u\>_{\Fq}\right\}\lneq \C.
      	\end{equation*}
      	since, by definition of the closure, we have that there exists at least one element $C$ of $\C$ such that $c^{(i)}\in\<u\>_{\Fq}$. This concludes the proof.
    \end{proof}
    
   In the next theorem we give a characterization of the closure-type dual anticodes. As a consequence, we have that the set of closure-tye dual anticodes $\omA^\cl$ is defined.
    
   \begin{theorem}
    \label{thm:dualAcl}
        Let $ A=\bigotimes_{i=1}^rA^{(i)}\in\A^\cl$. We have
        \begin{equation*}
            A^\perp=\sum_{i=1}^r\left(\bigotimes_{j=1}^{i-1}\Fq^{n_j}\right)\otimes \left(A^{(i)}\right)^\perp\otimes\left(\bigotimes_{j=i+1}^r\Fq^{n_j}\right).
        \end{equation*}
        In particular, $A^\perp$ is a perfect space.
    \end{theorem}
    \begin{proof}
        For ease of notation, in the remainder of this proof, we denote by $X_i$ and $Y_i$ the spaces 
        \begin{equation*}
            \left(\bigotimes_{j=1}^{i-1}\Fq^{n_j}\right)\otimes A^{(i)}\otimes\left(\bigotimes_{j=i+1}^r\Fq^{n_j}\right) \qquad\textup{ and }\qquad \left(\bigotimes_{j=1}^{i-1}\Fq^{n_j}\right)\otimes \left(A^{(i)}\right)^\perp\otimes\left(\bigotimes_{j=i+1}^r\Fq^{n_j}\right)
        \end{equation*}
        respectively, for any $i\in\{1,\ldots,r\}$. We claim that $X_i^\perp=Y_i$. It is immediate to check that $Y_i\subseteq X_i^\perp$. Moreover, by a dimension argument, we get
        \begin{equation*}
            \dimq{Y_i}=(n_i-\dimq{A_i})\prod_{\tiny\begin{matrix}j=1\\j\neq i\end{matrix}}^rn_j=n-\dimq{A_i}\prod_{\tiny\begin{matrix}j=1\\j\neq i\end{matrix}}^rn_j=n-\dimq{X_i}=\dimq{X_i^\perp}
        \end{equation*}
        where the last equality follows from  Lemma \ref{lem:propdual} \eqref{item2:propdual}. This implies the claim. Finally, since $ A=\bigcap_{i=1}^rX_i$, by Lemma \ref{lem:propdual} \eqref{item3:propdual}, we have $ A^\perp=\sum_{i=1}^r Y_i$. That $A^\perp$ is perfect follows immediately from the fact that $A^\perp$ is a sum of tensor spaces.  
    \end{proof}
    
    As an immediate consequence of this theorem we have the following characterization of the set $\overline{\A^\cl}$ of closure-type dual anticodes.
    
    \begin{corollary}
    \label{cor:barAcl}
        The set $\overline{\A^\cl}$ is defined. In particular, we have
        \begin{equation*}
            \overline{\A^\cl}=\left\{\sum_{i=1}^r\left(\bigotimes_{j=1}^{i-1}\Fq^{n_j}\right)\otimes A^{(i)}\otimes\left(\bigotimes_{j=i+1}^r\Fq^{n_j}\right):A^{(i)}\leq\Fq^{n_i}\textup{ for all } i \in \{1,\ldots,r\}\right\}.
        \end{equation*}
    \end{corollary}
    
    \begin{proposition}
       Let $A,B\in \overline{\A^\cl}$. Define  $A\vee_{\overline{\cl}} B$ to be $A+B$ and $A\wedge_{\overline{\cl}} B$ to be the sum of the elements in $\{C\in\B:C\in A\cap B\}$ where
       \begin{equation*}
           \B:=\left\{\left(\bigotimes_{j=1}^{i-1}\Fq^{n_j}\right)\otimes A^{(i)}\otimes\left(\bigotimes_{j=i+1}^r\Fq^{n_j}\right):A^{(i)}\leq\Fq^{n_i}\textup{ for all } i \in \{1,\ldots,r\}\right\}.
       \end{equation*}
       We have that $\overline{\A^\cl}$ with the operation of $\vee_{\overline{\cl}}$ and $\wedge_{\overline{\cl}}$ is a lattice.
    \end{proposition}
    \begin{proof}
       It is immediate from the definition that the operations in the statement are well-defined.
    \end{proof}
    
    The following example demonstrates that $\overline{\A^\cl}$ is not a sublattice of $\A^\ps$.
    \begin{example}
        Consider the following dual closure-type anticodes of $\F_2^3\otimes \F_2^3$
        
        \begin{equation*}
            A:=\<(0,1,0),(0,0,1)\>_{\F_2}\otimes\F_2^{3}+\F_2^{3}\otimes\<(1,0,0),(0,1,1)\>_{\F_2} \quad\textup{ and }\quad B:=\<(1,0,0)\>_{\F_2}\otimes\F_2^3.
        \end{equation*}
        
        One can check that
        \begin{equation*}
            A\wedge_{\ps} B=\<(1,0,0)\otimes(1,0,0),(1,0,0)\otimes (0,1,1)\>_{\F_2}\neq 0_{\ps}=0_{\overline{\cl}}= A\wedge_{\overline{\cl}} B
        \end{equation*}
        where $\wedge_{\ps}$ and $\wedge_{\overline{\cl}}$, and $0_{\ps}$ and $0_{\overline{\cl}}$ denote the operation of $\wedge$ and the zero element in the lattices $\A^\ps$ and $\overline{\A^\cl}$ respectively.
    \end{example}
    
    In the following result we give a characterization of the set $\A^\D$.
    
    \begin{theorem}
    \label{thm:charDelsarte}
       Let $P:=\{i\in\{1,\ldots,r\}:n_i=n_r\}$. We have 
       \begin{equation*}
           \A^\D=\bigcup_{p\in P}\left\{\left(\bigotimes_{j=1}^{i-1}\Fq^{n_j}\right)\otimes A^{(i)}\otimes\left(\bigotimes_{j=i+1}^r\Fq^{n_j}\right):A^{(i)}\leq\Fq^{n_i}, i \in \{1,\ldots,r\}\setminus\{p\} \right\}.
       \end{equation*}
    \end{theorem}
    \begin{proof}
        It was shown in Corollary \ref{cor:barAcl} that $\overline{\A^\cl}$ is a collection of anticodes and a characterization of it was given. Therefore, an equivalent definition of $\A^\D$ is
        \begin{equation*}
            \A^\D=\{A:A \in \A^\cl\cap\overline{\A^\cl} \textup{ and } A^{(p)}=\Fq^{n_p} \textup{ for some } p\in P\}.
        \end{equation*}
        It is not hard to check that, by the first part of the proof of Theorem~\ref{thm:dualAcl}, $\A^\cl\cap\overline{\A^\cl}$ is the set of all the perfect spaces of the form
        \begin{equation*}
             \left(\bigotimes_{j=1}^{i-1}\Fq^{n_j}\right)\otimes A^{(i)}\otimes\left(\bigotimes_{j=i+1}^r\Fq^{n_j}\right)
        \end{equation*}
        for some $A^{(i)}\leq\Fq^{n_i}$ and $i\in\{1,\ldots,r\}$. The statement now follows.
    \end{proof}
    
    As immediate consequence of this theorem we have the following corollary. In particular, it shows that the set of Delsarte-type dual anticodes is defined.
    
    \begin{corollary}
    \label{cor:charDeldual}
        $\A^\D$ is closed under duality, that is $\overline{\A^\D}=\A^\D$.
    \end{corollary}
    
    \begin{theorem}
    \label{thm:latticeDelsarte}
        Let $A,B\in\A^\D$. Define $A\vee_\D B$ to be $A+B$ and $A\wedge_\D B$ to be the Delsarte-type anticode of highest dimension contained in $A\wedge_\D B$. In particular, we have   
        \begin{equation*}
            A\wedge_\D B =\begin{cases}
                A\cap B & \textup{ if } A\cap B\in\A^\D,\\
                \{0\} & \textup{ otherwise}.
            \end{cases}
        \end{equation*}
        We have that $\A^\D$ with the operations of $\vee_\D$ and $\wedge_\D$ is a lattice.
    \end{theorem}
    \begin{proof}
        Let $A,B$ be elements of $\A^\D$. In particular, by Theorem \ref{thm:charDelsarte}, we have
        \begin{equation*}
            A=\left(\bigotimes_{s=1}^{i-1}\Fq^{n_s}\right)\otimes A^{(i)}\otimes\left(\bigotimes_{s=i+1}^r\Fq^{n_s}\right) \qquad\textup{ and }\qquad B=\left(\bigotimes_{s=1}^{j-1}\Fq^{n_s}\right)\otimes B^{(j)}\otimes\left(\bigotimes_{s=j+1}^r\Fq^{n_s}\right).
        \end{equation*}
        One can check that $A+B$ is a Delsarte-type anticode. Indeed, we have
        \begin{equation*}
            A+B=\begin{cases}
             \left(\bigotimes_{s=1}^{i-1}\Fq^{n_s}\right)\otimes \left(A^{(i)}+B^{(i)}\right)\otimes\left(\bigotimes_{s=i+1}^r\Fq^{n_s}\right) & \textup{ if } i=j,\\
                 \F & \textup{ otherwise.}
            \end{cases}
        \end{equation*}
        It remains to show that the operation of $\wedge_\D$ is well-defined. It is not hard to see that $A\cap B\in\A^\D$ if and only if $i=j$. On the other hand, if $i\neq j$, then we have
        
        \begin{equation}
        \label{eq:meetDelsarte}
            A\cap B=\left(\bigotimes_{s=1}^{i-1}\Fq^{n_s}\right)\otimes A^{(i)}\otimes\left(\bigotimes_{s=i+1}^{j-1}\Fq^{n_s}\right)\otimes B^{(j)}\otimes\left(\bigotimes_{s=j+1}^{r}\Fq^{n_s}\right)\notin\A^\D.
        \end{equation}
        Therefore, we have $A\wedge_{\D} B=0_{\A^\D}=\{0\}$ in $\A^\D$. The statement follows.
    \end{proof}
    
    \begin{remark}
        Observe that $\A^\D$ is not a sublattice of $\A^\cl$. In particular, we have the following. Let $A,B\in\A^\D$ be such that $A^{(i)}$ and $B^{(j)}$ are nontrivial spaces with $i\neq j$. Proposition \ref{thm:latticecl} and Equation \eqref{eq:meetDelsarte} imply $A\wedge_{\cl} B=A\cap B\neq 0_{\D}=A\wedge_{\D} B$.
    \end{remark}
    
    \begin{remark}
    \label{rem:ADr2}
    The theory of these anticodes reduces to the theory of the Delsarte optimal anticodes defined in \cite{ravagnani2016generalized} for $r=2$.
        Explicitly, observe that for $r=2$ we have
        \begin{equation*}
           \A^\D=
           \begin{cases}
               \left\{A^{(1)}\otimes\Fq^{n_2}:A^{(1)}\leq\Fq^{n_1}\right\} & \textup{ if } n_1 < n_2,\\[0.5em]
               \left\{A^{(1)}\otimes\Fq^{n_2}:A^{(1)}\leq\Fq^{n_1}\right\}\cup\left\{\Fq^{n_1}\otimes A^{(2)}:A^{(2)}\leq\Fq^{n_2}\right\} & \textup{ if } n_1=n_2.
           \end{cases}
       \end{equation*}
       This recovers the characterization of the Delsarte optimal   anticodes proved by Meshulam (see \cite[Theorem~3]{meshulam1985maximal}) for the case $n_1=n_2$ but from which the case $n_1<n_2$ easily follows.
    \end{remark}
    
    In the following result we give a characterization of the set $\A^\R$.
    
    \begin{theorem}
    \label{thm:charRav}
        Let $S:=\{i\in\{1,\ldots,r\}:n_i=n_1\}$. We have 
        \begin{equation*}
            \A^\R=\left\{\left(\bigotimes_{j=1}^{i-1}\Fq^{n_j}\right)\otimes A^{(i)}\otimes\left(\bigotimes_{j=i+1}^r\Fq^{n_j}\right):A^{(i)}\leq\Fq^{n_i}, i \in S \right\}.
        \end{equation*}
    \end{theorem}
    \begin{proof}
        By Corollary \ref{cor:barAcl}, an equivalent definition of $\A^\R$ is given by
        \begin{equation*}
            \A^\R=\{A:A\in\A^\cl\cap\overline{\A^\cl} \textup{ and } A^{(s)}=\Fq^{n_s} \textup{ for all } s\notin S\}.
        \end{equation*}
        The statement follows.
    \end{proof}
    
    The following corollary, which is a consequence of the theorem above, shows that the set of Ravagnani-type dual anticodes is defined.
    
    \begin{corollary}
        The set $\A^\R$ is closed under duality, that is $\overline{\A^\R}=\A^\R$.
    \end{corollary}
    
    \begin{corollary}
        $\A^\R$ is a sublattice of $\A^\D$.
    \end{corollary}
    \begin{proof}
        It follows immediately from Theorem \ref{thm:latticeDelsarte} and the facts that $A^\R\subseteq\A^\D$ and $0_\R=0_\D$, where $0_\R$ denotes the zero space in $\A^\R$.
    \end{proof}
    
    \begin{remark}
        It is not difficult to check that $\A^\R\subseteq\A^\D$. Moreover, observe that $\A^\R=\A^\D$ if $n_1=\cdots=n_r$ or $r=2$. The latter follows from the argument in Remark \ref{rem:ADr2}.
    \end{remark}
    
    Figures \ref{fig:r=2} (for $r=2$ or $n_1=\cdots=n_{r-1}$) and \ref{fig:r>2} (for $r>2$ and $\{n_1,\ldots,n_{r-1}\}\neq\{n_1\}$) give a set-theoretical representation of the families of anticodes defined above.

    \begin{minipage}{0.49\textwidth}
    \begin{figure}[H]
    \centering
    \resizebox{\textwidth}{!}{
    \begin{tikzpicture}[fill=gray]
    \draw (-0.8,-1) circle (3) (-1.5,2)  node [text=black,above] {\Large$\A^\cl$}
        (0.8,-1) circle (3) (1.5,2)  node [text=black,above] {\Large$\overline{\A^\cl}$};
    \draw[fill=mygreen!20](0,-1) circle (1.5) (0,-1.2)  node [text=black,above] {\Large$\A^\D=\A^\R$};
    \draw(-5.8,-4.5) rectangle (5.8,3) node[text=black,above] at (5.25,2.2) {\Large$\A^\ps$};
    \end{tikzpicture}
    }
    \vspace*{-0.5em}
    \caption{\label{fig:r=2}}
    \end{figure}
    \end{minipage}
    \begin{minipage}{0.49\textwidth}
        \begin{figure}[H]
    \centering
    \resizebox{\textwidth}{!}{
    \begin{tikzpicture}[fill=gray]
    \draw (-0.8,-1) circle (3) (-1.5,2)  node [text=black,above] {\Large$\A^\cl$}
        (0.8,-1) circle (3) (1.5,2)  node [text=black,above] {\Large$\overline{\A^\cl}$};
    \draw[fill=myyellow!20](0,-1.3) circle (1.8) (0,0.6)  node [text=black,above] {\Large$\A^\D$};
    \draw[fill=myblue!20](0,-1.3) circle (1) (0,-1.5)  node [text=black,above] {\Large$\A^\R$};
    \draw      (-5.8,-4.5) rectangle (5.8,3) node[text=black,above] at (5.25,2.2) {\Large$\A^\ps$};
    \end{tikzpicture}
    }
    \vspace*{-0.5em}
    \caption{\label{fig:r>2}}
    \end{figure}
    \end{minipage}

\section{Invariants for Tensor Anticodes}
\label{sec:invariants}
    
    In this section we define and study invariants for tensor anticodes and we compare them with the invariants for codes in the rank metric (see for example \cite[Section~6]{byrne2019tensor} and \cite[Sections~4,5]{ravagnani2016generalized}). If not explicitly stated,  definitions and results in the remainder hold for any collection of anticodes, possibly different from the ones identified in Section \ref{sec:anticodes}.
	
	\begin{definition}
	\label{def:tjsj}
		Let $\A$ be a collection of anticodes. For each $j\in\{1,\ldots,k\}$, the $j$\textbf{-th generalized tensor weight} with respect to $\A$ is defined to be:
		\begin{equation*}
			t_j(\C):=\min\left\{\dimq{A}:A \in\A\mid \dimq{\C \cap A}\geq j\right\}.
		\end{equation*}
		Furthermore, if $\omA \subseteq \A^\ps$ then we define the $j$\textbf{-th generalized dual tensor weight} to be
		\begin{equation*}
			s_j(\C):=\min\left\{\dimq{A}:A \in\omA\mid \dimq{\C \cap A}\geq j\right\}.
		\end{equation*}
	\end{definition}
	
	In the following, we write $t_j$, $s_j$, $t_j^\perp$ and $s_j^\perp$ instead of $t_j(\C)$, $s_j(\C)$, $t_j\left(\C^\perp\right)$ and $s_j\left(\C^\perp\right)$ respectively, for any $j\in\{1,\ldots,k\}$. As we observed in the previous section, in some cases $\omA$ is a collection of anticodes, while in others it is not. Particular formulations of the following result for $r=2$ can be read in \cite[Proposition~6.3]{byrne2019tensor} and \cite[Theorem~30]{ravagnani2016generalized}. These statements follow directly from Definitions \ref{anticodedef} and \ref{def:tjsj}.
	
	\begin{proposition}
	\label{thm:propt}
		The following hold.
		\begin{enumerate}
		\setlength\itemsep{0.5em}
			\item $d\leq t_1$ and if $\omA\subseteq\A^\ps$  then $d\leq s_1$.
			\item $ \trk(\C)\leq t_k$ and if $\omA\subseteq\A^\ps$ then $\trk(\C)\leq s_k$.
			\item \label{item3:propt} $t_j\leq t_{j+1}$ for any $j\in\{1,\ldots, k\}$.
			\item $s_j\leq s_{j+1}$ for any $j\in\{1,\ldots, k\}$ if $\omA\subseteq\A^\ps$.
			\item $t_j^\ps\leq t_j^\cl\leq t_j^\D\leq t_j^\R$ for any $j\in\{1,\ldots, k\}$ .
			\item $t_j^\D=s_j^\D$ and $t_j^\R=s_j^\R$ for any $j\in\{1,\ldots, k\}$.
			\item $s_j^\cl\leq t_j^\D\leq t_j^\R$ for any $j\in\{1,\ldots, k\}$.
		\end{enumerate}
	\end{proposition}
	
	The following result summarizes the properties of the generalized tensor weights for perfect spaces. We omit the proof of this result, which is similar to the proof of \cite[Theorem~6.3]{byrne2019tensor}.
		
	\begin{proposition}
	\label{prop:proptps}
		The following hold.
		\begin{enumerate}
		\setlength\itemsep{0.5em}
            \item\label{item1:proptps} $t_1^\ps=d$.
            \item $t_k^\ps=\trk(\C)$.
            \item $t_j^\ps<t_{j+1}^\ps$.
            \item\label{item4:proptrk} $t_j^\ps\leq \trk(\C)+k-j$.
            \item\label{item5:proptrk} $t_j^\ps\geq d+j-1$.
        \end{enumerate}
	\end{proposition}
	
	One can consider \eqref{item4:proptrk} and \eqref{item5:proptrk} to be  generalizations of Kruskal's tensor rank bound \cite[Corollary~1]{kruskal1977three}. In particular, for $r=2$, it is possible to recover the latter from \eqref{item4:proptrk} with $j=1$ and from \eqref{item5:proptrk} with $j=k$. In the following, we denote by $d_j(\C)$ the $j$-th Delsarte generalized rank weight of $\C$, for any $j\in\{1,\ldots,k\}$ and $r=2$, as defined in \cite{ravagnani2016generalized}. In the following, we write $d_j$ instead of $d_j(\C)$, for ease of notation. The next result shows that the closure-type anticodes allow us to define new invariants for linear spaces of 2-tensors (i.e. rank-metric codes) that are refinements of the Delsarte generalized rank weights.
    
    \begin{proposition}
    \label{prop:propcl2}
       The following hold for $r=2$.
       \begin{enumerate}
       \setlength\itemsep{0.5em}
           \item\label{item1:propcl2} $t_1^\cl=d^2$.
           \item\label{item2:propcl2} $t_k^\cl=\dimq{\cl(\C)}$.
           \item\label{item3:propcl2} $t_j^\cl,s_j^\cl\leq d_j\cdot n_2$ for all $j\in\{1,\ldots,k\}$.
       \end{enumerate}
    \end{proposition}
    \begin{proof}
     For any $C\in \C$ we have $\cl(C)=\colsp(C)\otimes\rowsp(C)$ and hence $\dimq{\cl(C)}=\rk(C)^2$. It follows that
            \begin{equation*}
                t_1^\cl=\min\left\{\dimq{\cl(C)}:C \in \C\setminus\{0\}\right\}=\left(\min\{\rk(C):C\in\C\setminus\{0\}\}\right)^2=d^2.
            \end{equation*}
            This proves \eqref{item1:propcl2}. Clearly, (\ref{item2:propcl2}) holds since $\cl(\C)$ is the smallest space in $\A^\cl$ that contains $\C$.
            
            Finally, \eqref{item3:propcl2} follows from the fact that, for $r=2$, the set of Delsarte anticodes defined in \cite{ravagnani2016generalized} is a subset of $\A^\cl\cap\overline{\A^\cl}$ by Remark \ref{rem:ADr2}.
        This concludes the proof.
    \end{proof}

    \begin{remark}
        Observe that $s_1^\cl$ is not a multiple in $d$ in general. For example, let $\C\leq\F_3^3\otimes\F_3^3$ be the code generated by $X:=(1,0,0)\otimes(0,1,0)+(0,1,0)\otimes(1,0,0)$. Clearly the minimum distance of $\C$ is $2$. One can check that $s_1(\C)=5$ and a closure-type dual anticode of dimension $5$ containing $X$ is $A:=\<(1,0,0)\>_{\F_3}\otimes\F_3^3+\F_3^3\otimes\<(1,0,0)\>_{\F_3}$.
    \end{remark}
    
    \begin{remark}
        The value of $t_j^\cl$ and $s_j^\cl$, $j\in\{1,\ldots,k\}$ are, in general, not comparable. Consider, for example, the matrices 
        \begin{equation*}
            M:=\begin{pmatrix}
            1&0&0\\0&0&0\\0&0&0
            \end{pmatrix} \qquad \textup{ and }\qquad N:=\begin{pmatrix}
            0&0&1\\0&1&0\\1&0&0
            \end{pmatrix}
        \end{equation*}
        over $\F_2$. One can observe the following.
        \begin{itemize}
            \item The smallest closure-type anticode and dual closure-type anticode that contain $M$ are respectively $\<(1,0,0)\>_{\F_2}\otimes\<(1,0,0)\>_{\F_2}$ and $\<(1,0,0)\>_{\F_2}\otimes\F_2^3$. This implies $$t_j^\cl\left(\<M\>_{\F_2}\right)=1<3=s_j\left(\<M\>_{\F_2}\right).$$
            \item The smallest closure-type anticode and dual anticode that contain $N$ are respectively $\F_2^3\otimes\F_2^3$, i.e. the full-space, and $\<(1,0,0),(0,1,0)\>_{\F_2}\otimes\F_2^3+\F_2^3\otimes\<(1,0,0)\>_{\F_2}$. This implies $$t_j^\cl\left(\<N\>_{\F_2}\right)=9>7=s_j\left(\<N\>_{\F_2}\right).$$
        \end{itemize}
    \end{remark}
    
    In the following example, we compute the generalized tensor weights associated with $\A^\cl$ of the code in Example \ref{ex:closure}.
    
    \begin{example}
        Let $\C\leq \F_3^2\otimes \F_3^3\otimes \F_3^4$ be the code as in Example \ref{ex:closure}. One can check that
        \begin{itemize}
            \item $t_1^\cl(\C)=4$, which is obtained for example for the closure of $X$;
            \item $t_2^\cl(\C)=18$, which is obtained for the closure of $\<X,Y\>_{\F_3}$, that is $$\F_3^2\otimes\F_3^3\otimes\<(1,0,0,1),(0,1,0,0),(0,0,1,2)\>_{\F_3};$$
            \item $t_3^\cl(\C)=24$ which implies that $\F_3^2\otimes\F_3^3\otimes\F_3^4$ is the smallest closure-type anticode containing $\C$.
        \end{itemize}
    \end{example}

    In the following example, we compare the generalized tensor weights corresponding to the various closure-type anticodes we have introduced. In particular, we demonstrate that in some cases, the closure-type anticodes and their duals provide invariants that distinguish inequivalent codes, while the generalized rank weights associated to the Delsarte anticodes of \cite{ravagnani2016generalized} do not. We briefly recall the definitions of \textit{MRD} and \textit{dually QMRD} matrix codes. These are codes of maximum cardinality.
    
    \begin{remark}
    \label{rem:MRD}
       A code  $\C\leq\Fq^{n_1}\otimes\Fq^{n_2}$ that satisfies
        \begin{equation*}
            d=n_1-\left\lfloor\frac{k-1}{n_2}\right\rfloor \qquad\textup{and}\qquad d^\perp=\left\lceil\frac{k+1}{n_2}\right\rceil
        \end{equation*}
        is said to be \textbf{MRD} (\textbf{Maximum Rank Distance}) if $n_2\mid k$ and \textbf{dually QMRD} (\textbf{Quasi Maximum Rank Distance}) if $n_2\nmid k$. These codes exist for any value of $n_1,n_2$ and $k$ (see \cite{de2018weight,delsarte1978bilinear} for further details).
    \end{remark}
    
    \begin{example}
    \label{ex:Gabidulin}
        Consider the following inequivalent $4$-dimensional codes over $\F_3$:
    \begin{align*}
        \C:=&\<
        \begin{pmatrix}
            1&0&0&0\\0&1&0&0
        \end{pmatrix},
        \begin{pmatrix}
            0&1&0&0\\0&0&1&0
        \end{pmatrix},
        \begin{pmatrix}
            0&0&1&0\\0&0&0&1
        \end{pmatrix},
        \begin{pmatrix}
            0&0&0&1\\1&0&0&1
        \end{pmatrix}
        \>_{\F_3},\\
        \mD:=&\<
        \begin{pmatrix}
            1&0&0&0\\2&0&1&1
        \end{pmatrix},
        \begin{pmatrix}
            0&1&0&0\\1&2&0&2
        \end{pmatrix},
        \begin{pmatrix}
            0&0&1&0\\2&1&2&2
        \end{pmatrix},
        \begin{pmatrix}
            0&0&0&1\\2&2&1&1
        \end{pmatrix}
        \>_{\F_3}.
    \end{align*}
    In particular, $\C$ and $\mD$ are $1$-dimensional Delsarte-Gabidulin codes and in particular, they are $\Fq$-$[2\times 4,4,2]$ MRD codes. One can check the following.
    
    \begin{itemize}
        \item We have $t_1^\ps(\C)=t_1^\ps(\mD)=2$. This value is obtained, for examples, for the codeword 
        \begin{equation*}
            \begin{pmatrix}
            1&0&0&0\\0&1&0&0
            \end{pmatrix}=\begin{pmatrix}
            1&0&0&0\\0&0&0&0
            \end{pmatrix}+\begin{pmatrix}
            0&0&0&0\\0&1&0&0
            \end{pmatrix}
        \end{equation*}
        of $\C$ and for the codeword
        \begin{equation*}
            \begin{pmatrix}
            1&0&0&0\\2&0&1&1
            \end{pmatrix}=\begin{pmatrix}
            1&0&0&0\\0&0&0&0
            \end{pmatrix}+\begin{pmatrix}
            0&0&0&0\\2&0&1&1
            \end{pmatrix}
        \end{equation*}
        of $\mD$. This is in line with Proposition \ref{prop:proptps} \eqref{item1:proptps}.
        
        \item We have $t_1^\cl(\C)=s_1^\cl(\C)=4$. This value is obtained, for example, for the perfect space $\F_3^2\otimes\<(1,0,0,0),(0,1,0,0)\>_{\F_3}$. Observe that this latter is a closure-type anticode and dual anticode whose intersection with $\C$ is the span of
        \begin{equation*}
            \begin{pmatrix}
            1&0&0&0\\0&1&0&0
            \end{pmatrix}.
        \end{equation*}
        This is in line with Proposition \ref{prop:propcl2}\ref{item1:propcl2} as we have $t_1^\cl(\C)=d(\C)^2$.
        
        \item We have $t_2^\cl(\C)=s_2^\cl(\C)=6$. This value is obtained, for example, for the perfect space $\F_3^2\otimes\<(1,0,0,0),(0,1,0,0),(0,0,1,0)\>_{\F_3}$.  Observe that this latter is a closure-type anticode and dual anticode whose intersection with $\C$ is the $2$-dimensional subspace generated by
        \begin{equation*}
            \begin{pmatrix}
            1&0&0&0\\0&1&0&0
        \end{pmatrix}\qquad\textup{ and }\qquad
        \begin{pmatrix}
            0&1&0&0\\0&0&1&0
        \end{pmatrix}.
        \end{equation*}
        
        \item We have $t_1^\cl(\mD)=t_2^\cl(\mD)=s_1^\cl(\mD)=s_2^\cl(\mD)=4$. This value is obtained, for example, for the perfect space $\F_3^2\otimes\<(1,0,0,0),(0,0,1,1)\>_{\F_3}$. Observe that this latter is a closure-type anticode and dual anticode whose intersection with $\mD$ is the $2$-dimensional subspace generated by
        \begin{equation*}
            \begin{pmatrix}
            1&0&0&0\\2&0&1&1
            \end{pmatrix}\qquad \textup{ and }\qquad \begin{pmatrix}
            0&0&1&1\\1&0&0&0
            \end{pmatrix}.
        \end{equation*}
        Observe that this is in line with Proposition \ref{prop:propcl2}\ref{item1:propcl2} as we have $t_1^\cl(\mD)=d(\mD)^2$.
        
        \item We have $s_3^\cl(\C)=s_3^\cl(\mD)=7$. This value is obtained, for example, for the closure-type dual anticode $\F_3^2\otimes \<(1,0,0,0),(0,0,1,0),(0,0,0,1)\>_{\F_3}+\<(1,0),(0,1,0,0)\>_{\F_3}$. The intersection of the latter with $\C$ is the span of
        \begin{equation*}
            \begin{pmatrix}
            0&1&0&0\\0&0&1&0
        \end{pmatrix},
        \begin{pmatrix}
            0&0&1&0\\0&0&0&1
        \end{pmatrix},
        \begin{pmatrix}
            0&0&0&1\\1&0&0&1
        \end{pmatrix}
        \end{equation*}
        and with $\mD$ is the span of
        \begin{equation*}
            \begin{pmatrix}
            1&0&0&0\\2&0&1&1
        \end{pmatrix},
        \begin{pmatrix}
            0&1&0&2\\2&0&2&1
        \end{pmatrix},
        \begin{pmatrix}
            0&0&1&1\\1&0&0&0
        \end{pmatrix}.
        \end{equation*}
        
        \item We have $t_3^\cl(\C)=t_4^\cl(\C)=s_4^\cl(\C)=8$ and $t_3^\cl(\mD)=t_4^\cl(\mD)=s_4^\cl(\mD)=8$. This means that the full-space is the only closure-type anticode that intersects $\C$ and $\mD$ in a subspace of dimension at least $3$, and the only closure-type dual anticode that contains the codes $\C$ and $\mD$.
        
        \item We have $d_j(\C)=d_j(\mD)=2$ for all $j\in\{1,\ldots,4\}$. This shows that the Delsarte generalized rank weights do not distinguish $\C$ and $\mD$.
    \end{itemize}
    \end{example}
    
    The following result extends \cite[Theorem~30, (4)]{ravagnani2016generalized} for $r>2$. The proof is similar and we include it for completeness. 
    
    \begin{proposition}
    \label{prop:proptR}
        For any $j\in\{1,\ldots,k\}$ the following hold. 
        \begin{enumerate}
        \setlength\itemsep{0.5em}
            \item\label{item1:proptR}  $t_j^\R +\frac{n}{n_1}\leq t_{j+\frac{n}{n_1}}^\R$. 
            \item $t_j^\R\leq n-\frac{n}{n_1}\left\lfloor\frac{n_1}{n}(k-j)\right\rfloor $. 
            \item\label{item3:proptR} $t_j^\R\geq j$.
        \end{enumerate}
    \end{proposition}
    \begin{proof}
        \begin{enumerate}
            \item  We first claim that for any non-trivial $A\in\A^\R$  there exists $\hat A\in\A^\R$ such that $\hat A\leq A$ and $\dimq{\hat A}=\dimq{A}-\frac{n}{n_1}$. Let $\displaystyle A=A^{(1)}\otimes\bigotimes_{i=2}^r\Fq^{n_i}$. Let $\hat A^{(1)}\leq A^{(1)}$ be such that $\dimq{\hat A^{(1)}}=\dimq{A^{(1)}}-1$. We have 
            \begin{equation*}
                \dimq{\hat A^{(1)}\otimes \bigotimes_{i=2}^r\Fq^{n_i}}=\left(\dimq{A^{(1)}}-1\right)\prod_{i=2}^r\dimq{\Fq^{n_i}}=\dimq{A}-\frac{n}{n_1}.
            \end{equation*}
            Now let $A\in\A^\R$ be such that $\dimq{A}=t_{j+\frac{n}{n_1}}^\R$ and $\dimq{A\cap\C}\geq j+\frac{n}{n_1}$. Let $\hat A$ be an   anticode of dimension $\dimq{A}-\frac{n}{n_1}$ and such that $\hat A\leq A$. We want to show that $\dimq{\hat A\cap\C}\geq j$. Since $\hat A\leq A$ we have $\hat A\cap\C=\hat A\cap(A\cap\C)$. Therefore,
            \begin{align*}
                \dimq{\hat A\cap \C}&=\dimq{\hat A\cap(A\cap\C)}\\
                &=\dimq{\hat A}+\dimq{A\cap\C}-\dimq{\hat A+(A\cap\C)}\\
                &\geq \dimq{\hat A}+\dimq{A\cap\C}-\dimq{A}\\
                &\geq \dimq{A}-\frac{n}{n_1} + j+\frac{n}{n_1}-\dimq{A},
            \end{align*}
            which implies the statement.
        
            \item Let $h:=\left\lfloor\frac{n_1}{n}(k-j)\right\rfloor$. It follows from \eqref{item1:proptR} that $t_j^\R+\ell \frac{n}{n_1}\leq t_{j+\ell \frac{n}{n_1}}^\R $ for each $\ell \in [h]$ and therefore $t_j^\R+h \frac{n}{n_1}\leq n$. 
            The statement now follows.
            \item If $A\in\A^\R$ is such that $\dimq{A\cap\C}\geq j$ then $\dimq{A}\geq j$ and so \eqref{item3:proptR} holds.\qedhere
        \end{enumerate}
    \end{proof}
    
    We conclude this section by providing a generalization, for $r>0$, of the Wei-type duality described in \cite[Section~6]{ravagnani2016generalized}. We omit the proof of the following results as they are similar to the proofs of \cite[Theorem~37, Corollary~38]{ravagnani2016generalized}.
    \begin{lemma}
        Let $\C$ be a code such that $1\leq \dimq{\C}\leq n-1$. Let $p,i,j\in\Z$ be such that $1\leq p+i\frac{n}{n_1}\leq n-k$ and $1\leq p+k+j\frac{n}{n_1}\leq k$. We have
        $$t_{p+i\frac{n}{n_1}}^\R(\C^\perp)\neq n+\frac{n}{n_1}-t_{p+k+j\frac{n}{n_1}}^\R(\C).$$
    \end{lemma}
    
    \begin{theorem}
        Define the sets
        \begin{align*}
            S_p(\C^\perp)&:=\left\{\frac{n_1}{n}\,t_{p+i\frac{n}{n_1}}^\R(\C^\perp):i\in\Z \textup{ and } 1\leq p+i\frac{n}{n_1}\leq k\right\},\\ \overline{S}_p(\C)&:=\left\{n_1+1-\frac{n_1}{n}\,t_{p+i\frac{n}{n_1}}^\R(\C):i\in\Z \textup{ and }1\leq p+i\frac{n}{n_1}\leq k\right\}.
        \end{align*}
        We have $S_p(\C^\perp)=\{1,\ldots,n_1\}\setminus\overline{S}_p(\C)$ for any $1\leq p\leq \frac{n}{n_1}$. In particular, the generalized tensor weights (of the Ravagnani-type) of $\C$ fully determine those of $\C^\perp$. 
    \end{theorem}

\section{Binomial Moments and Weight Distributions}

In this section, we define and study other invariants associated to the different collection of anticodes we identified, namely the \textit{generalized binomial moments} and the \textit{generalized weight distribution}. We also comment on how these invariants are related to those ones introduced in \cite{byrne2020rank}. Throughout this section we use the convention that $\sum_{x\in\emptyset}f(x)=0$ for any function $f:\Q\longrightarrow\Q$.

	\begin{definition}
        Let $a\in\{0,\ldots,n\}$ and $j\in\{1,\ldots,k\}$. The $(a,j)$\textbf{-th generalized tensor binomial moment} of $\C$ is defined to be
        \begin{equation*}
            B_a^{(j)}(\C):=\sum_{A\in\A_a}B_A^{(j)}(\C),\qquad\textup{ where }\qquad B_A^{(j)}(\C):=\qbin{\dimq{\C\cap A}}{j}
        \end{equation*}
        for any $A\in \A$. If moreover $\omA\subseteq\A^\ps$ we define the $(a,j)$\textbf{-th generalized dual tensor binomial moment} of $\C$ to be
        \begin{equation*}
            \oB_a^{(j)}(\C):=\sum_{A\in\omA_a}\oB_{A}^{(j)}(\C),\qquad\textup{ where }\qquad \oB_{A}^{(j)}(\C):=\qbin{\dimq{\C\cap A}}{j}.
        \end{equation*}
  	\end{definition}
  
  	\begin{definition}
        Let $j\in\{1,\ldots,k\}$. The $j$\textbf{-th generalized weight distribution} is defined to be the vector of length $n$ whose $a$-th component is given by
        \begin{equation*}
            W_a^{(j)}(\C):=\sum_{A\in\A_a}W_A^{(j)}(\C),
        \end{equation*}
        where $W_A^{(j)}:=|\{\mD\leq(\C\cap A):\dimq{\mD}=j \textup{ and } A=\bigwedge\{B\in\A_a:\mD\leq B\}\}|$, for any $A\in\A_a$. Moreover if  $\omA\subseteq\A^\ps$ we let the $j$\textbf{-th generalized dual weight distribution} to be the vector of length $n$ whose $a$-th component is given by
        \begin{equation*}
            \oW_a^{(j)}(\C):=\sum_{A\in\omA_a}\oW_A^{(j)}(\C),
        \end{equation*}
		where $\oW_A^{(j)}:=|\{\mD\leq(\C\cap A):\dimq{\mD}=j \textup{ and } A=\bigwedge\{B\in\omA_a:\mD\leq B\}\}|$, for any $A\in\omA_a$.
  	\end{definition}
  	
  	\begin{definition}
        Let $X$ and $Y$ be indeterminates. We define the $j$\textbf{-th tensor weight enumerator} by
        \begin{equation*}
            \W_\C^{(j)}(X,Y):=\sum_{a\leq n}W_a^{(j)}(\C)X^{n-a}Y^a\in \Q[X,Y],
        \end{equation*}
        for all $j\in\{1,\ldots,k\}$. Moreover if  $\omA\subseteq\A^\ps$ we let the $j$\textbf{-th dual tensor weight enumerator} be defined by
        \begin{equation*}
            \omW_\C^{(j)}(X,Y):=\sum_{a\leq n}\oW_a^{(j)}(\C)X^{n-a}Y^a\in \Q[X,Y],
        \end{equation*}
        for all $j\in\{1,\ldots,k\}$.
    \end{definition}
  	
  	A particular formulation of the following result for $r=2$ and $\A=\A^\D$ is given by \cite[Proposition~1]{blanco2018rank} and \cite[Corollary~33]{ravagnani2016rank} for $j=1$, and \cite[Theorem~3.8]{byrne2020rank} for any $j\in\{1,\ldots,k\}$. Recall that each collection of anticodes identified in Section \ref{sec:anticodes} is a lattice.
  	
  	\begin{theorem}
  	\label{thm:BaWb}
  	    Let $\A$ be a family of anticodes whose members form a lattice and let $\mu$ be its M\"obius function.
  		The following hold for all $a,b\in\{0,\ldots,n\}$ and $j\in\{1,\ldots,k\}$.
  		\begin{enumerate}
  			\item $\displaystyle B_a^{(j)}(\C)=\sum_{b=0}^aW_b^{(j)}(\C)\;\left|\left\{(A',A)\in\A_b\times\A_a:A'\leq A\right\}\right|$.
            \item $\displaystyle W_{b}^{(j)}(\C)=\sum_{a=0}^b\mu(a,b)B_a^{(j)}(\C)\;\left|\left\{(A',A)\in\A_b\times\A_a:A\leq A'\right\}\right|$.
  		\end{enumerate}
  	\end{theorem}
  	\begin{proof}
        For ease of notation in the remainder of this proof, we write $B_a^{(j)}$ and $W_b^{(j)}$ instead of $B_a^{(j)}(\C)$ and $W_b^{(j)}(\C)$. First of all, observe that
        \begin{equation}
        \label{eq1:Moebius}
            \sum_{\tiny\begin{matrix}A'\in\A\\A'\leq A\end{matrix}} W_{A'}^{(j)}=|\{\mD\leq\C:\dimq{\mD}=j\textup{ and } \mD\leq A\}|=B_A^{(j)}.  
        \end{equation}
        Therefore, by M\"obius inversion formula we get
        \begin{equation}
        \label{eq2:Moebius}
            W_{A'}^{(j)}=\sum_{\tiny\begin{matrix}A\in\A\\A\leq A'\end{matrix}} \mu(\dimq{A},\dimq{A'})B_A^{(j)}.
        \end{equation}
        We prove the two statement independently.
        \begin{enumerate}
            \item By Equation \eqref{eq1:Moebius}, we have
                \begin{align*}
                    B_a^{(j)}&=\sum_{A\in\A_a}B_A^{(j)}=\sum_{A\in\A_a}\sum_{\tiny\begin{matrix}A'\in\A\\A'\leq A\end{matrix}} W_{A'}^{(j)}\\
                    &=\sum_{b=0}^a\;\sum_{A'\in\A_{b}}W_{A'}^{(j)}|\{A\in\A_a,A'\leq A\}|\\
                    &=\sum_{b=0}^aW_b^{(j)}\left|\left\{(A',A)\in\A_b\times\A_a:A'\leq A\right\}\right|.
                \end{align*}
            \item By Equations \eqref{eq1:Moebius} and \eqref{eq2:Moebius}, we have
                \begin{align*}
                    W_{b}^{(j)}&=\sum_{A'\in\A_b}W_{A'}^{(j)}=\sum_{A'\in\A_b}\sum_{\tiny\begin{matrix}A\in\A\\A\leq A'\end{matrix}} \mu(\dimq{A},b)B_A^{(j)}\\
                    &=\sum_{a=0}^b\mu(a,b)\sum_{A\in\A_a}B_A^{(j)}|\{A\in\A_b,A'\leq A\}|\\
                    &=\sum_{a=0}^b\mu(a,b)B_a^{(j)}\left|\left\{(A',A)\in\A_b\times\A_a:A\leq A'\right\}\right|.\qedhere
                \end{align*}
        \end{enumerate}
  	\end{proof}
  	
  	Particular formulations for the case $r=2$ and $\A=\A^\D$ of the next two result can be read in \cite[Lemma~28]{ravagnani2016rank} and \cite[Lemma~3]{byrne2020rank}, respectively.
  	
  	\begin{lemma}
    \label{lem:dualMatU}
        For any $A\in\A_a$ we have $\dimq{\C\cap A}=\dimq{\C^\perp\cap A^\perp}+k+a-n$.
    \end{lemma}
    \begin{proof}
        By Lemma \ref{lem:propdual} we have
        \begin{align*}
            \dimq{\C\cap A}&=n-\dimq{(\C\cap A)^\perp}\\
            &=n-\dimq{\C^\perp+A^\perp}\\
            &=n-\dimq{\C^\perp}-\dimq{A^\perp}+\dimq{\C^\perp\cap A^\perp}\\
            &=\dimq{\C^\perp\cap A^\perp}+k+a-n.\qedhere
        \end{align*}
    \end{proof}
   
    \begin{theorem}
    \label{thm:Bj}
    	Let $a\in\{0,\ldots,n\}$ and suppose that $\omA\subseteq\A^\ps$. For any $j\in\{1,\ldots,k\}$, we have that
    	\begin{equation*}
    		B_a^{(j)}(\C)=
    		\begin{cases}
    			0 & \textup{ if } a < t_j,\\[1ex]
    			\qbin{k+a-n}{j}\left|\A_a\right| & \textup{ if } a > n-s_1^\perp.
    		\end{cases}
    	\end{equation*}  
    \end{theorem}
    \begin{proof}
    	Observe first that if $a< t_j$ then for all $A \in\omA$ we have
    	\begin{equation*}
    		\qbin{\dimq{\C\cap A}}{j}=0.
    	\end{equation*}
    	On the other hand, if $n-a<s_1^\perp$ then $\C^\perp \cap A^\perp=\{0\}$, for all $A\in\A_a$. Therefore, Lemma \ref{lem:dualMatU} implies
    	\begin{equation*}
    		B_a^{(j)}(\C)=\sum_{A\in\A_a}\qbin{\dimq{\C\cap A}}{j}=\sum_{A\in\A_a}\qbin{k+a-n}{j}=\qbin{k+a-n}{j}|\A_a|.
    	\end{equation*}
    	This concludes the proof.
    \end{proof}
    
    The following result establish the MacWilliams identities for generalized tensor binomial moments. A particular formulation for $r=2$ and $\A=\A^\D$ of the following result is \cite[Theorem~7.1]{byrne2020rank}.
  	
  	\begin{theorem}
  	\label{thm:dualBj}
  		The following holds for any $j\in\{1,\ldots,k\}$ and $a\in\{0,\ldots,n\}$, under the assumption that  $\omA\subseteq\A^\ps$.
  		\begin{equation*}  
  			B_a^{(j)}(\C)=\sum_{p=0}^jq^{p(k+a-n-j+p)}\qbin{k+a-n}{j-p}\oB_{n-a}^{(p)}(\C^\perp).
  		\end{equation*}
  	\end{theorem}
  	\begin{proof}
  		Lemmas \ref{lem:bin} and \ref{lem:dualMatU} imply
  		\begin{align*}
  			B_a^{(j)}(\C)=&\sum_{A\in\A_a}\qbin{\dimq{\C^\perp\cap A^\perp}+k+a-n}{j}\\
  			&=\sum_{A\in\A_a}\sum_{p=0}^jq^{p(k+a-n-j+p)}\qbin{k+a-n}{j-p}\qbin{\dimq{\C^\perp\cap A^\perp}}{p}\\
  			&=\sum_{p=0}^jq^{p(k+a-n-j+p)}\qbin{k+a-n}{j-p}\sum_{A\in\A_a}\qbin{\dimq{\C^\perp\cap A^\perp}}{p}\\
  			&=\sum_{p=0}^jq^{p(k+a-n-j+p)}\qbin{k+a-n}{j-p}\sum_{A\in\omA_{n-a}}\qbin{\dimq{\C^\perp\cap A}}{p},
  		\end{align*}
  		where the last equality follows from the fact that the map $A\longmapsto A^\perp$ is a bijection between the set of   tensor anticodes of dimension $a$ and the set of   dual tensor anticodes of dimension $n-a$. This implies the statement.
  	\end{proof}
  
  	\begin{remark}
  	    MacWilliams identities for generalized tensor weight distributions can be easily derived by applying  Theorem \ref{thm:BaWb} to Theorem \ref{thm:dualBj}. A particular formulation for $r=2$ and $\A=\A^\D$ is \cite[Corollary~7.2]{byrne2020rank}.
  	\end{remark}
  	
    We now introduce a class codes whose generalized tensor binomial moments are determined by their code parameters
    $[n_1 \times \cdots \times n_r,k,d]$. This in turn implies that their generalized tensor weight distributions are also determined. 
    
    \begin{definition}
    \label{def:TBMD}
        Suppose that  $\omA\subseteq\A^\ps$ . For any $j\in\{1,\ldots,k\}$, we say that the code $\C$ is $j$\textbf{-TBMD} (\textbf{Tensor Binomial Moment Determined}) with respect to (the anticodes in)  $\A$ if $n-s_1^\perp-t_j<0$. Moreover, we say that $\C$ is \textbf{minimally} $j$\textbf{-TBMD} if $j$ is the minimum of the set $\{p\in\{1,\ldots,k\}:\C \textup{ is } j\textup{-TBMD}\}$.
    \end{definition}
    
    The following is an immediate consequence of Theorem \ref{thm:propt} \eqref{item3:propt}.
    \begin{proposition}
       Suppose that $\omA\subseteq\A^\ps$. If $\C$ is $j$-TBMD with respect to  $\A$ then $\C$ is $(j+1)$-TBMD with respect to  $\A$.
    \end{proposition}
    
    \begin{proposition}
         Suppose that  $\omA\subseteq\A^\ps$. If $\C$ is $j$-TBMD with respect to  $\A$ then its $j$-th tensor binomial moments and tensor weight distribution are determined by its code parameters.
    \end{proposition}
    \begin{proof}
       If $\C$ is $j$-TBMD then $n-s_1^\perp<t_j$ and the result follows from Theorems \ref{thm:BaWb} and \ref{thm:Bj}.
    \end{proof}
    
    As immediate consequence of Proposition \ref{thm:propt} we have the following.
    
    \begin{proposition}
    \label{prop:jTBMD}
        The following hold for any $j\in\{1,\ldots,k\}$.
        \begin{enumerate}
        \setlength\itemsep{0.5em}
            \item If $\C$ is $j$-TBMD with respect to $\A^\cl$ or $\omA^\cl$ then it is $j$-TBMD with respect to $\A^\D$.
            \item If $\C$ is $j$-TBMD with respect to $\A^\D$ then it is $j$-TBMD with respect to $\A^\R$.
            \item\label{item3:jTBMD} $\C$ is $1$-TBMD with respect to $\A^\D$ (or $\A^\R$) if and only if $\C^\perp$ is $1$-TBMD with respect to $\A^\D$ (or $\A^\R$ respectively).
        \end{enumerate}
    \end{proposition}
    
    \begin{remark}
        It follows from \ref{prop:jTBMD} \eqref{item3:jTBMD} that the notion of extremality associated with $\A^\D$ and $\A^\R$ is invariant under duality. 
    \end{remark}
    
    We recall the following definition of extremal matrix codes introduces in \cite{byrne2020rank}.
    
    \begin{definition}[{\cite[Definition~4.1]{byrne2020rank}}]
        For any $j\in\{1,\ldots,k\}$, we say that the code of $2$-tensors $\C\leq\Fq^{n_1}\otimes\Fq^{n_2}$ is $j$-BMD if $n_1-d_j-d^\perp<0$.
    \end{definition}

    As a consequence of Remark \ref{rem:ADr2}, we have that the notion of extremality defined above coincides with the one in Definition \ref{def:TBMD},  for $r=2$ and $\A=\A^\D$ (and therefore $\A=\A^\R$).
    
    \begin{proposition}
    \label{prop:jBMD}
       Let $r=2$ and $j\in\{1,\ldots,k\}$. $\C$ is $j$-TBMD (with respect to  $\A^\D$ or $\A^\R$) if and only if $\C$ is $j$-BMD.
    \end{proposition}
    \begin{proof}
       Let $\C$ be $j$-TBMD. Remark \ref{rem:ADr2} implies
       \begin{equation*}
           n_1\cdot n_2-(s_1^\D)^\perp-t_j^\D= n_1\cdot n_2-(t_1^\D)^\perp-t_j^\D= n_2\,(n_1-d^\perp-d_j).
       \end{equation*}
       The statement follows from the definitions of $j$-TBMD and $j$-BMD code.
    \end{proof}
    
    As observed in \cite{byrne2020rank}, the class of $1$-BMD codes is partitioned by the families of MRD and dually QMRD codes that exist for any value of $n_1,n_2$ and $k$. Therefore, we have that $1$-TBMD codes with respect to $\A^\D$ exist for $r=2$.
    
    In the following example we demonstrate again that invariants associated with the closure-type anticodes distinguish codes even if those associated with the Delsarte-type anticodes do not. In particular, in this example we will present a pair of Delsarte-Gabidulin codes which are both minimally $1$-TBMD with respect to the Delsarte-type anticodes. However, these codes are minimally $j$-TBMD with respect to the closure-type anticodes for different values of $j$. 
    
    \begin{example}
        Let $\C$ and $\mD$ be the Delsarte-Gabidulin codes as in Example \ref{ex:Gabidulin}. One can check that $s_1^\cl(\C^\perp)=4$. This value is attained, for example, for the dual closure-type anticode $\F_3^2\otimes\<(1,0,2,0),(0,1,2,0)\>_{\F_3}$, which intersects $\C^\perp$ in the space spanned by
        \begin{equation*}
            \begin{pmatrix}
            1&2&0&0\\0&2&1&0
        \end{pmatrix}.
        \end{equation*}
        Therefore we have $ 8-s_1^\cl(\C^\perp)-t_1^\cl(\C)=8-4-4=0\not<0$ which implies that $\C$ is not $1$-TBMD with respect to the closure-type anticodes. On the other hand, $\C$ is minimally $2$-TBMD with respect to the closure-type anticodes since $8-s_1^\cl(\C^\perp)-t_2^\cl(\C)=8-4-6=-2<0$. Moreover, one can also verify that $s_1^\cl(\mD^\perp)=4$. This value is attained, for example, by the dual closure-type anticode $\F_3^2\otimes\<(1,0,0,1),(0,0,1,0)\>_{\F_3}$ which intersects $\mD^\perp$ in the space spanned by
        \begin{equation*}
            \begin{pmatrix}
            1&0&0&1\\2&0&2&2
            \end{pmatrix}\qquad\textup{ and }\qquad
            \begin{pmatrix}
            0&0&1&0\\2&0&0&2
            \end{pmatrix}.
        \end{equation*}
        Therefore we have $ 8-s_1^\cl(\mD^\perp)-t_1^\cl(\mD)=8-4-4=0\not<0$ which implies that $\C$ is not $1$-TBMD with respect to the closure-type anticodes. One can also observe that $\mD$ is not $2$-TBMD either but it is minimally $3$-TBMD since $8-s_1^\cl(\mD^\perp)-t_3^\cl(\mD)=8-4-8=-4<0$. Finally, we have that  $t_1^\D(\C)=t_1^\D(\C^\perp)=t_1^\D(\mD)=t_1^\D(\mD^\perp)=8$ (recall that for any tensor code $\C$, $s_j(\C)=t_j(\C)$, for all $j\in\{1,\ldots,\dimq{\C}\}$, since $\A^\D=\omA^\D$). This is in line with Remark \ref{rem:ADr2} as every column-space of each codeword in $\C$, $\C^\perp$, $\mD$ and $\mD^\perp$ is the full space $\F_3^2$. Therefore, we have that both $\C$ and $\mD$ are $1$-TBMD with respect to the Delsarte anticodes since $8-t_1^\D(\C^\perp)-t_1^\D(\C)=8-8-8=-8<0$. This is in line with Proposition \ref{prop:jBMD} and \cite[Remark~4.17]{byrne2020rank} as the Delsarte-Gabidulin codes are MRD.
    \end{example}
    
    We conclude this section with two examples in which we study the generalized tensor weights of two tensor codes using the construction given in \cite{roth1996tensor}. These classes of tensor codes can be seen as the generalisation of the Delsarte-Gabidulin codes for $r>2$. It has been shown that all the MRD codes (i.e. the case $r=2$) are minimally $1$-TBMD with respect to the Delsarte-type anticodes\footnote{ Recall that for the case $r=2$ the Delsarte-type and Ravagnani-type anticodes coincide.}  (see \cite[Remark~4.17]{byrne2020rank} and Proposition \ref{prop:jBMD}). However, the following examples show that not all the codes constructed as in \cite{roth1996tensor} are minimally $1$-TBMD with respect to the Delsarte-type and Ravagnani-type anticodes. This is perhaps not surprising in light of the fact that a strict improvement on the Singleton bound can be obtained (see \cite{roth1991maximum}).
    
    \begin{example}
    \label{ex:roth}
    Let $\alpha$ be a primitive element of $\F_{4}$ and define the vectors $\boldsymbol{\alpha}:=(1,\alpha)$, $\boldsymbol{\beta}:=(1,\alpha)$ and $\boldsymbol{\omega}:=(1,\alpha)$ whose components are linearly independent over $\F_2$. Let $\mS(2,3,3;2)$ be the set of all pairs $(\ell,s)$ such that $\ell,s\in\{0,\ldots, 1\}$ and there exist an $\F_2$-$[2,\ell+1,\geq s+1]$ linear block code. It is not hard to see that $\mS(2,3,3;2)=\{(0,0),(0,1),(1,0)\}$.  Moreover, $\overline{\mS}(2,3,3;2)$ be the set of all pairs $(\ell,s)$ such that $\ell,s\in\{0,\ldots, 1\}$ and $(\ell,s)\notin  \mS(2,3,3;2)$, i.e. $\overline{\mS}(2,3,3;2)=\{(1,1)\}$.  Define, as in \cite[Sections~3.1 and 3.2]{roth1996tensor}, the matrices with entries in $\F_{4}$
    \begin{align*}
        \mH_{(\ell,s),(i,j)}&:=\left(\boldsymbol{\alpha}_i\right)^{2^\ell}\left(\boldsymbol{\beta}_j\right)^{2^s} \qquad \textup{ with } (\ell,s)\in \mS(2,3,3;2),\\
        \mG_{(\ell,s),(i,j)}&:=\left(\boldsymbol{\alpha}_i^\perp\right)^{2^\ell}\left(\boldsymbol{\beta}_j^\perp\right)^{2^s} \qquad \textup{ with } (\ell,s)\in \overline{\mS}(2,3,3;2), 
    \end{align*}
    where $\boldsymbol{\alpha}^\perp$ and $\boldsymbol{\beta}^\perp$ are the dual bases of $\boldsymbol{\alpha}$ and $\boldsymbol{\beta}$ respectively.  In particular, we have that $\boldsymbol{\alpha}^\perp=\boldsymbol{\beta}^\perp=(\alpha^2,1)$. One can check that
    \begin{equation*}
        \mH=\begin{pmatrix}
            1 & \alpha & \alpha & \alpha^2\\
            1 & \alpha^2 & \alpha & 1\\
            1 & \alpha & \alpha^2 & 1
        \end{pmatrix}\qquad \textup{ and }\qquad 
        \mG= \begin{pmatrix}
            \alpha^2 & \alpha & \alpha & 1
        \end{pmatrix}.
    \end{equation*}
     As described in \cite[Section~3.2]{roth1996tensor},  $\mG$ can be seen as the ``generator matrix'' and $\mH$ as the ``parity-check matrix'' of the tensor code $\C(2,3,3;2)\leq\F_2^2\otimes\F_2^2\otimes\F_2^2$. More in detail, the code $\C(2,3,3;2)$ is defined as the set of all tensors $X\in\F_2^2\otimes\F_2^2\otimes\F_2^2$ such that
     \begin{equation*}
         \sum_{i,j,t=1}^2X_{i,j,t}\left(\boldsymbol{\alpha}_i\right)^{2^\ell}\left(\boldsymbol{\beta}_j\right)^{2^s}\boldsymbol{\omega}_t=0 \qquad \textup{ for all } (\ell,s)\in \mS(2,3,3;2). 
     \end{equation*}
     A $3$-tensor is obtained from $\mG$ as follows. First express $\mG$ as a $2\times 2$ matrix over $\F_4$ whose first row is $(\alpha^2,\alpha)$ and whose second row is $(\alpha,1)$. Next express each coefficients as binary vectors of length two whose first coefficient belongs the first slice and whose second coefficient belongs to the second one. In this way we obtain the tensor $C_1$ shown below. $C_2$ is obtained in the same way from $\alpha\,\mG$. Therefore, we have that $\C(2,3,3;2)$ is the $2$-dimensional tensor code generated by
     \begin{equation*}
         C_1:=\left(\begin{array}{cc|cc}
            1 & 0 & 1 & 1  \\
            0 & 1 & 1 & 0
         \end{array}\right)
         \qquad\textup{ and }\qquad
         C_2:=\left(\begin{array}{cc|cc}
            1 & 1 & 0 & 1  \\
            1 & 0 & 1 & 1
         \end{array}\right),
     \end{equation*}
     and has the minimum distance at least $3$, by \cite[Theorem~4]{roth1996tensor}. One can check the following.
    \begin{itemize}
        \item We have $t_1^\ps=3$. For example, we have that $C_1$ can be written as sum of $3$ simple tensors as follows,
        \begin{equation*}
            \left(\begin{array}{cc|cc}
            1 & 0 & 1 & 1  \\
            0 & 1 & 1 & 0
         \end{array}\right)=
         \left(\begin{array}{cc|cc}
            1 & 1 & 1 & 1  \\
            0 & 0 & 0 & 0
         \end{array}\right)+
          \left(\begin{array}{cc|cc}
            0 & 0 & 0 & 0  \\
            0 & 0 & 1 & 0
         \end{array}\right)+
          \left(\begin{array}{cc|cc}
            0 & 1 & 0 & 0  \\
            0 & 1 & 0 & 0
         \end{array}\right).
        \end{equation*}
        \item We have $t_1^\cl=t_1^\D=t_1^\R=8$. The  value $t_1^\cl$ is obtained for the closure-type anticode $\F_2^2\otimes\F_2^2\otimes\F_2^2$ which clearly meets the code in a space of dimension at least one. As an immediate consequence we have also $t_2^\cl=t_2^\D=t_2^\R=8$.
        \item We have $(s_1^\cl)^\perp=(s_1^\D)^\perp=(s_1^\R)^\perp=4$. This value is obtained, for example for the dual closure-type anticode $\F_2^2\otimes\F_2^2\otimes \<(0,1)\>_{\F_2}$ whose intersection with the dual code is the span of
        \begin{equation*}
            \left(\begin{array}{cc|cc}
            0 & 0 & 1 & 0  \\
            0 & 0 & 1 & 1
         \end{array}\right).
        \end{equation*}
    \end{itemize}
    In particular, we have $2^3-(s_1^\cl)^\perp-t_1^\cl=8-4-8=-4<0$ which implies that $\C(2,3,3;2)$ is $1$-TBMD with respect to the closure-type anticodes. It follows that  $\C(2,3,3;2)$ is also $1$-TBMD with respect to the Delsarte-type and Ravagnani-type anticodes.
\end{example}

\begin{example}
    Let $\alpha$ be a root of $x^3+x+1\in\F_2[x]$ and define the vectors $\boldsymbol{\alpha}:=(1,\alpha,\alpha^2)$, $\boldsymbol{\beta}:=(1,\alpha,\alpha^2)$ and $\boldsymbol{\omega}:=(1,\alpha,\alpha^2)$ whose components are linearly independent over $\F_2$. Let $\mS(3,4,3;2)$ be the set of all pairs $(\ell,s)$ such that $\ell,s\in\{0,\ldots, 2\}$ and there exist an $\F_2$-$[3,\ell+1,\geq s+1]$ linear block code. It is not hard to check that 
    \begin{equation*}
        \mS(3,4,3;2)=\{(0,0),(0,1),(0,2),(1,0),(1,1),(2,0)\}.
    \end{equation*}
    Moreover, $\overline{\mS}(3,4,3;2)$ be the set of all pairs $(\ell,s)$ such that $\ell,s\in\{0,\ldots, 2\}$ and $(\ell,s)\notin  \mS(3,4,3;2)$. In this case we have
    \begin{equation*}
        \overline{\mS}(3,4,3;2)=\{(1,2),(2,1),(2,2)\}.
    \end{equation*}
    Define, as in \cite[Sections~3.1 and 3.2]{roth1996tensor}, the matrices with entries in $\F_8$
    \begin{align*}
        \mH_{(\ell,s),(i,j)}&:=\left(\boldsymbol{\alpha}_i\right)^{2^\ell}\left(\boldsymbol{\beta}_j\right)^{2^s} \qquad \textup{ with } (\ell,s)\in \mS(3,4,3;2),\\
        \mG_{(\ell,s),(i,j)}&:=\left(\boldsymbol{\alpha}_i^\perp\right)^{2^\ell}\left(\boldsymbol{\beta}_j^\perp\right)^{2^s} \qquad \textup{ with } (\ell,s)\in \overline{\mS}(3,4,3;2), 
    \end{align*}
    where $\boldsymbol{\alpha}^\perp$ and $\boldsymbol{\beta}^\perp$ are the dual bases of $\boldsymbol{\alpha}$ and $\boldsymbol{\beta}$ respectively. In particular, we have $\boldsymbol{\alpha}^\perp=\boldsymbol{\beta}^\perp=(1,\alpha^2,\alpha)$
    One can check that
    \begin{align*}
        \mH=\begin{pmatrix}
            1 &   \alpha & \alpha^2 &   \alpha & \alpha^2 & \alpha^3 & \alpha^2 & \alpha^3 & \alpha^4\\
            1 & \alpha^2 & \alpha^4 &   \alpha & \alpha^3 & \alpha^5 & \alpha^2 & \alpha^4 & \alpha^6\\
            1 & \alpha^4 &   \alpha &   \alpha & \alpha^5 & \alpha^2 & \alpha^2 & \alpha^6 & \alpha^3\\
            1 &   \alpha & \alpha^2 & \alpha^2 & \alpha^3 & \alpha^4 & \alpha^4 & \alpha^5 & \alpha^6\\
            1 & \alpha^2 & \alpha^4 & \alpha^2 & \alpha^4 & \alpha^6 & \alpha^4 & \alpha^6 &   \alpha\\
            1 &   \alpha & \alpha^2 & \alpha^4 & \alpha^5 & \alpha^6 &   \alpha & \alpha^2 & \alpha^3
        \end{pmatrix},
        \\[1ex]
       \mG=\begin{pmatrix}
            1 &   \alpha & \alpha^4 & \alpha^4 & \alpha^5 &   \alpha & \alpha^2 & \alpha^3 & \alpha^6\\
            1 & \alpha^4 & \alpha^2 &   \alpha & \alpha^5 & \alpha^3 & \alpha^4 &   \alpha & \alpha^6\\
            1 &   \alpha & \alpha^4 &   \alpha & \alpha^2 & \alpha^5 & \alpha^4 & \alpha^5 &   \alpha
        \end{pmatrix}.
    \end{align*}
    
    As described in \cite[Section~3.2]{roth1996tensor},  $\mG$ can be seen as the ``generator matrix'' and $\mH$ as the ``parity-check matrix'' of the tensor code $\C(3,4,3;2)\leq\F_2^3\otimes\F_2^3\otimes\F_2^3 $. More in detail, the code $\C(3,4,3;2)$ is defined as the set of all tensors $X\in\F_2^2\otimes\F_2^2\otimes\F_2^2$ such that
    
     \begin{equation*}
         \sum_{i,j,t=1}^3X_{i,j,t}\left(\boldsymbol{\alpha}_i\right)^{2^\ell}\left(\boldsymbol{\beta}_j\right)^{2^s}\boldsymbol{\omega}_t=0 \qquad \textup{ for all } (\ell,s)\in \mS(3,4,3;2). 
     \end{equation*} 
     
     As in Example \ref{ex:roth} expanding each coordinate with respect to the basis $\boldsymbol{\omega}$ we get that $\C(3,4,3;2)$ is the $9$-dimensional tensor code generated by
     
    \begin{align*}
C_1&:=\left(\begin{array}{ccc|ccc|ccc}
1 & 0 & 0 & 0 & 1 & 1 & 0 & 0 & 1\\
0 & 1 & 0 & 1 & 1 & 1 & 1 & 1 & 0\\
0 & 1 & 1 & 0 & 1 & 0 & 1 & 0 & 1
\end{array}\right),
&C_2&:=\left(\begin{array}{ccc|ccc|ccc}
0 & 0 & 1 & 1 & 0 & 1 & 0 & 1 & 1\\
1 & 1 & 0 & 1 & 0 & 0 & 1 & 1 & 1\\
1 & 0 & 1 & 1 & 1 & 0 & 0 & 1 & 0\\
\end{array}\right),\\[1ex]
C_3&:=\left(\begin{array}{ccc|ccc|ccc}
0 & 1 & 1 & 0 & 1 & 0 & 1 & 0 & 1\\
1 & 1 & 1 & 0 & 0 & 1 & 1 & 0 & 0\\
0 & 1 & 0 & 1 & 1 & 1 & 1 & 1 & 0\\
\end{array}\right),
&C_4&:=\left(\begin{array}{ccc|ccc|ccc}
1 & 0 & 0 & 0 & 1 & 0 & 0 & 1 & 1\\
0 & 1 & 1 & 1 & 1 & 1 & 0 & 1 & 0\\
0 & 0 & 1 & 1 & 1 & 0 & 1 & 0 & 1\\
\end{array}\right),\\[1ex]
C_5&:=\left(\begin{array}{ccc|ccc|ccc}
0 & 1 & 1 & 1 & 1 & 1 & 0 & 1 & 0\\
0 & 1 & 0 & 0 & 0 & 1 & 1 & 1 & 1\\
1 & 0 & 1 & 1 & 0 & 0 & 1 & 1 & 0\\
\end{array}\right),
&C_6&:=\left(\begin{array}{ccc|ccc|ccc}
0 & 1 & 0 & 0 & 0 & 1 & 1 & 1 & 1\\
1 & 1 & 1 & 1 & 0 & 1 & 0 & 0 & 1\\
1 & 1 & 0 & 0 & 1 & 1 & 1 & 0 & 0\\
\end{array}\right),\\[1ex]
C_7&:=\left(\begin{array}{ccc|ccc|ccc}
1 & 0 & 0 & 0 & 1 & 1 & 0 & 0 & 1\\
0 & 0 & 1 & 1 & 0 & 1 & 0 & 1 & 1\\
0 & 1 & 0 & 1 & 1 & 1 & 1 & 1 & 0\\
\end{array}\right),
&C_8&:=\left(\begin{array}{ccc|ccc|ccc}
0 & 0 & 1 & 1 & 0 & 1 & 0 & 1 & 1\\
0 & 1 & 1 & 0 & 1 & 0 & 1 & 0 & 1\\
1 & 1 & 0 & 1 & 0 & 0 & 1 & 1 & 1\\
\end{array}\right),\\[1ex]
C_9&:=\left(\begin{array}{ccc|ccc|ccc}
0 & 1 & 1 & 0 & 1 & 0 & 1 & 0 & 1\\
1 & 0 & 1 & 1 & 1 & 0 & 0 & 1 & 0\\
1 & 1 & 1 & 0 & 0 & 1 & 1 & 0 & 0\\
\end{array}\right),
    \end{align*}
   and has the minimum distance at least $4$, by \cite[Theorem~4]{roth1996tensor} and computations show that its minimum distance is at most $5$. Moreover, one can check the following.
    \begin{itemize}
        \item We have $t_1^\cl=t_1^\D=t_1^\R=18$. This value is obtained, for example, for the closure-type anticode $\<(0,1,0),(0,0,1)\>_{\F_2}\otimes\F_2^3\otimes\F_2^3$ whose intersection with the code is the span of
        \begin{equation*}
\left(\begin{array}{ccc|ccc|ccc}
0 & 0 & 0 & 0 & 0 & 0 & 0 & 0 & 0\\
0 & 1 & 0 & 1 & 1 & 1 & 1 & 1 & 0\\
1 & 0 & 1 & 1 & 1 & 0 & 0 & 1 & 0
\end{array}\right).
        \end{equation*}
        
        \item We have $(s_1^\cl)^\perp=(s_1^\D)^\perp=(s_1^\R)^\perp=9$. This value is obtained, for example, for the dual closure-type anticode $\F_2^3\otimes\<(1,1,0)\>_{\F_2}\otimes\F_2^3$ whose intersection with the code is the span of
        \begin{equation*}
\left(\begin{array}{ccc|ccc|ccc}
0 & 0 & 0 & 1 & 1 & 0 & 1 & 1 & 0\\
1 & 1 & 0 & 1 & 1 & 0 & 1 & 1 & 0\\
1 & 1 & 0 & 0 & 0 & 0 & 1 & 1 & 0
\end{array}\right).
    \end{equation*}
    
    \end{itemize}
    In particular, we have $3^3-(s_1^\R)^\perp-t_1^\R=3^3-(t_1^\R)^\perp-t_1^\R=27-9-18=0\not<0$ which implies that $\C(3,4,3;2)$ is not $1$-TBMD with respect to the Ravagnani-type anticodes. As a consequence $\C(3,4,3;2)$ is not $1$-TBMD with respect to the closure-type and Delsarte-type anticodes either.
\end{example}

\vspace{1cm}
\section*{Appendix}

The following table includes the characterization of the four collections of anticodes described in Section \ref{sec:anticodes}.
\vspace{1em}
    \renewcommand\arraystretch{2}
    \begin{longtable}[c]{@{}|C{2.7cm}|C{12.15cm}|@{}}
    \hline
  	Anticodes & Characterization\\\specialrule{.13em}{.0em}{.0em} 
  	Perfect Spaces &  $\displaystyle\A^{\cl}=\{A\leq\F:A \textup{ is perfect}\}$\\\hline
  	\multirow{3.4}{*}{Closure-type}  & \multirow{1.7}{*}{$\displaystyle\A^\cl=\left\{\bigotimes_{i=1}^rA^{(i)}:A^{(i)}\leq\Fq^{n_i}\;\forall\, i \in [r]\right\}$}\\[-1.7ex]&\\\cline{2-2}
  	& \multirow{2}{*}{$\displaystyle\overline{\A^\cl}=\left\{\sum_{i=1}^r\left(\bigotimes_{j=1}^{i-1}\Fq^{n_j}\right)\otimes A^{(i)}\otimes\left(\bigotimes_{j=i+1}^r\Fq^{n_j}\right):A^{(i)}\leq\Fq^{n_i}\;\forall\, i \in [r]\right\}$}\\&\\\hline
  	\multirow{2.7}{*}{Delsarte-type} &  \multirow{2}{*}{$\displaystyle\A^\D=\bigcup_{p\in P}\left\{\left(\bigotimes_{j=1}^{i-1}\Fq^{n_j}\right)\otimes A^{(i)}\otimes\left(\bigotimes_{j=i+1}^r\Fq^{n_j}\right):A^{(i)}\leq\Fq^{n_i}, i \in [r]\setminus\{p\} \right\}$} \\&\\[-2ex]
  	& \raggedright\arraybackslash where $P:=\{i\in[r]:n_i=n_r\}$\\\hline
  	\multirow{2.5}{*}{Ravagnani-type} &  \multirow{2}{*}{$\displaystyle\A^\R=\left\{\left(\bigotimes_{j=1}^{i-1}\Fq^{n_j}\right)\otimes A^{(i)}\otimes\left(\bigotimes_{j=i+1}^r\Fq^{n_j}\right):A^{(i)}\leq\Fq^{n_i}, i \in S \right\}$}\\&\\[-2ex]
  	& \raggedright\arraybackslash where $S:=\{i\in[r]:n_i=n_1\}$\\\hline
  	\caption{\hspace{-0.3em}: Anticodes and their characterization.}
  	\end{longtable}

In the following table we summarize the main properties of the families of anticodes considered in Section \ref{sec:anticodes}.
\vspace{1em}

 \renewcommand\arraystretch{2}
\begin{longtable}[c]{@{}|C{2.6cm}|C{2.4cm}|C{1.75cm}|C{1.75cm}|C{2.4cm}|C{2.4cm}|@{}}
\hline
  	  & $\A^\ps$ & $\A^\cl$ & $\omA^\cl$ & $\A^\D$ & $\A^\R$\\\specialrule{.13em}{.0em}{.0em}
  	  Isometry invariance\footnote{ We refer to the isometries defined in Notation \ref{notation:isom}.} & YES & YES & YES & YES & YES \\\hline
  	  $\A=\omA$ & Not defined\footnote{\label{footnote:dualps} The set of dual anticodes is not defined for the family of perfect space (see Example \ref{ex:dualperfsp}).} & NO & NO & YES & YES \\\hline
  	  Is it a lattice? & YES & YES\footnote{ It is a sublattice of $\A^\ps$.} & YES & YES & YES\footnote{ It is a sublattice of $\A^\D$.} \\\hline
  	  Duality of $\{t_j:j \in [k]\}$ & Not defined\footref{footnote:dualps} & Unknown & Unknown &  Incomplete\footnote{ So far, we have a duality theory only for $r=2$ as $\A^\D=\A^\R$ in this case.} & YES\\\hline
  	  Duality of $\big\{B_u^{(j)}:j \in [k]\big\}$ &  Not defined\footref{footnote:dualps} & YES\footnote{\label{footnote:Bu} We have relations between $B_u^{(i)}$s and $\oB_u^{(i)}$s. For the family of closure-type anticodes, these quantities are not equal in general.} & YES\footref{footnote:Bu} & YES & YES\\\hline
  	  Invariants extended for $r>2$ & Generalized tensor ranks of \cite{byrne2019tensor}\footnote{  Extended in Proposition~\ref{prop:proptps}} & & & Generalized rank weights of \cite{ravagnani2016generalized}\footnote{\label{footnote:rankw} Extended in Proposition~\ref{prop:proptR}} &  Generalized rank weights of \cite{ravagnani2016generalized}\footref{footnote:rankw}\\\hline
  	  \caption{\hspace{-0.3em}: Properties of anticodes.}
\end{longtable}

\newpage
\bibliographystyle{abbrv} 
\bibliography{psbiblio.bib}
\end{document}